\documentclass{article}
\usepackage[english]{babel}

\usepackage[letterpaper,top=2cm,bottom=2cm,left=3cm,right=3cm,marginparwidth=1.75cm]{geometry}

\usepackage{authblk}
\usepackage{amsmath}
\usepackage{amsfonts}
\usepackage{amssymb}
\usepackage{amsthm}
\usepackage{graphicx}
\usepackage{xcolor}
\usepackage[framemethod=tikz]{mdframed}
\usepackage[colorlinks=true, allcolors=blue]{hyperref}
\usepackage{algorithm}
\usepackage{algpseudocode}

\usepackage{soul}
\usepackage{etoolbox}
\usepackage{standalone}
\usepackage{euscript}
\usepackage[numbers,sort&compress]{natbib}

\usepackage{pgfplots}
\usepackage{mathtools}

\usepackage{tikz}
\usepackage{tikzsymbols}
\usepackage{tikz-3dplot}
\usetikzlibrary{decorations.markings}
\usetikzlibrary{decorations.pathreplacing}
\usetikzlibrary{arrows.meta}
\usetikzlibrary{shapes}
\usetikzlibrary{calc}
\usetikzlibrary{math}
\tikzset{>=latex}

\newtheorem{theorem}{Theorem}

\newtheorem*{theorem*}{Theorem}
\newtheorem{lemma}[theorem]{Lemma}
\newtheorem{proposition}[theorem]{Proposition}

\newtheorem{corollary}[theorem]{Corollary}

\newtheorem{remark}[theorem]{Remark}
\newtheorem{assumption}[theorem]{Assumption}
\newtheorem{definition}[theorem]{Definition}

\usepackage{cleveref}

\newcommand{\rk}[0]{\operatorname{rank}}
\newcommand{\Cay}{\operatorname{Cay}}
\newcommand{\res}[2]{\left. #1 \right|_{#2}}
\newcommand{\col}{\operatorname{col}}
\newcommand{\bb}{\mathbb}
\newcommand{\mc}{\mathcal}
\newcommand{\eu}{\EuScript}
\newcommand{\spn}{\operatorname{span}}

\newtoggle{figures}
\toggletrue{figures}

\begin{document}
\title{An efficient decoder for a linear distance quantum LDPC code}
\author[1]{Shouzhen Gu}
\author[1]{Christopher A. Pattison}
\author[2]{Eugene Tang}

\affil[1]{\small Institute for Quantum Information and Matter, California Institute of Technology, Pasadena, CA 91125}
\affil[2]{\small Center for Theoretical Physics, Massachusetts Institute of Technology, Cambridge, MA 02139}

\date{\today}

\maketitle

\begin{abstract}
    Recent developments have shown the existence of quantum low-density parity check (qLDPC) codes with constant rate and linear distance. A natural question concerns the efficient decodability of these codes. In this paper, we present a linear time decoder for the recent quantum Tanner codes construction of asymptotically good qLDPC codes, which can correct all errors of weight up to a constant fraction of the blocklength. Our decoder is an iterative algorithm which searches for corrections within constant-sized regions. At each step, the corrections are found by reducing a locally defined and efficiently computable cost function which serves as a proxy for the weight of the remaining error.
\end{abstract}

\section{Introduction}

Quantum error correcting codes with constant-sized check operators, known as quantum low-density parity check (qLDPC) codes, have myriad applications in computer science and quantum information. Indeed, almost all leading contenders~\cite{DKLP02,BM06} for experimentally realizable fault-tolerant quantum memories are qLDPC codes. With more stringent requirements on their parameters, qLDPC codes can be used to achieve constant overhead fault-tolerant quantum computation as shown by Gottesman~\cite{G13}. On the more theoretical side, qLDPC codes are believed to have connections to the quantum probabilistically checkable proofs (qPCP) conjecture~\cite{EH15}.

A qLDPC code of blocklength $n$ is said to be good when it encodes $\Theta(n)$ logical qubits and detects all errors up to weight $\Theta(n)$. For many years such codes have proven elusive, with an apparent distance ``barrier'' of around $\sqrt{n}$. It is natural to wonder if there is some fundamental limitation that prevents us from achieving the {\it a priori} best possible distance of $\Theta(n)$. However, a sequence of recent constructions of qLDPC codes with steadily improving code parameters~\cite{PK20,HHO21,BE21} have culminated in the construction of asymptotically good qLDPC codes by Panteleev and Kalachev~\cite{PK22}. Alternative constructions of good qLDPC codes have since been given by Leverrier and Z{\'e}mor~\cite{LZ22} and conjectured by Lin and Hsieh~\cite{LH22quantum}.

With the proven existence of good qLDPC codes, a natural next step is to better understand their properties. For fault-tolerance purposes, a fast decoder is a necessity, so an important question is whether these codes can be efficiently decoded. Previously known efficient decoders~\cite{LTZ15, DLB21, EKZ20, PK19, QC21} were limited by the parameters of the underlying qLDPC code. To date, the best efficient decoder corrects against all errors of weight up to $\Theta(\sqrt{n}\log n)$~\cite{EKZ20}. The existence of good qLDPC codes opens the possibility for a decoder that corrects all errors of weight up to $\Theta(n)$.

In this paper, we focus on the quantum Tanner codes construction of Leverrier and Z{\'e}mor~\cite{LZ22}. Quantum Tanner codes were inspired by the original construction of good qLDPC codes of Panteleev and Kalachev~\cite{PK22}, as well as by the classical locally testable codes of Dinur, {\it et al.}~\cite{Dinur21}, serving as a intermediary between the two constructions. They can also be seen as a natural quantum generalization of classical Tanner codes~\cite{SS96}. A classical Tanner code is defined by placing bits on the edges of an expanding graph, with non-trivial checks defining local codes placed at the vertices. The codewords are the strings whose local views at each vertex belong to the codespace of the local code. A quantum Tanner code is a Calderbank-Shor-Steane (CSS)~\cite{CS96,Steane96} code defined by two classical Tanner codes stitched together using a two-dimensional expanding complex. For particular choices of the local checks and expanding complex, this construction has been shown to yield an asymptotically good family of qLDPC codes. We show that this construction can also yield an asymptotically good family of qLDPC codes which are efficiently decodable for errors of weight up to a constant fraction of the distance.

Our decoder is inspired by the small-set-flip~\cite{LTZ15} decoding algorithm for hypergraph product codes based on expanding graphs. Small-set-flip is an iterative algorithm, where at every step, small sets of qubits are flipped to decrease the syndrome weight. The candidate sets to flip are contained within the supports of individual stabilizer generators. A critical ingredient in the success of the small-set-flip decoder is the presence of expansion in the underlying geometric complex. Since the geometric complex defining quantum Tanner codes has a similar notion of expansion, one might expect that analogous ideas may work for decoding quantum Tanner codes.

In our decoder, we define a ``local potential function'' on each local view which measures the distance of the error from the local codespace. The decoder reduces the sum of these potential functions by applying a constant-sized correction within some local view at each step. In the proof of correctness, we proceed by tracking the minimum weight correction according to each local view, and then use this data to show that a flip-set with the required properties must exist when the error is not too large. As a required step in the proof, we also strengthen the robustness parameters of the random classical codes used in the quantum Tanner code construction.

Our main result is stated below: 
\begin{theorem*}[Informal version of Theorems~\ref{thm:adversarial_errors} and~\ref{thm:linear_time}]
There exists a family of asymptotically good quantum Tanner codes such that our decoder successfully corrects all errors of weight up to $\Theta(n)$ and runs in time $O(n)$.
\end{theorem*}

The remainder of the paper is organized as follows. In Section~\ref{sec:quantum_tanner} we provide a brief technical introduction to the quantum Tanner codes construction of asymptotically good qLDPC codes. There we present a terse, but self-contained, description of all the ingredients necessary to follow the rest of the paper. In Section~\ref{sec:decoder} we formally define the decoding problem and present the overview of our decoder for the quantum Tanner codes. We also work out basic properties and consequences of our decoder in this section. Section~\ref{sec:proof} contains the technical bulk of the paper, and presents the main proof of the correctness of the decoder. Finally, in Section~\ref{sec:conclusion} we provide a summary of our results and conclude with some open problems. We also include a technical appendix detailing the existence of the dual tensor codes with sufficiently high robustness parameter ($\Delta^{3/2+\varepsilon}$) necessary for the proof.

\section{Quantum Tanner Codes}\label{sec:quantum_tanner}

In this section, we review some coding theory background and summarize the construction of quantum Tanner codes by Leverrier and Z\'emor~\cite{LZ22}.

\subsection{Classical linear codes}
In this subsection we quickly review the necessary classical coding background. A classical linear code is a $k$-dimensional subspace $C\subseteq \bb F_2^n$, which is often specified by a parity check matrix $H\in \bb F_2^{(n-k)\times n}$ such that $C=\ker H$. Equivalently, the code can also be specified as the column space of a generator matrix $G\in \bb F_2^{n\times k}$, such that $C=\col G$. The parameter $n$ is called the blocklength of the code. The number of encoded bits is $k$ and $\rho=k/n$ is the rate of the code. The number of errors that the code can correct is determined by the distance of $C$, which is given by the minimum Hamming weight of a nonzero codeword: $d=\min_{x\in C\setminus\{0\}}|x|$. Sometimes, we consider the relative distance $\delta=d/n$. We say that such a code has parameters $[n,k,d]$.

Given a $D$-regular (multi)graph $\mc G=(V,E)$ and a code $C_0$ of blocklength $D$, we can define the classical Tanner code $C=T(\mc G,C_0)$ as follows. The bits of $C$ are placed on the edges of $\mc G$, so it is a code of length $n=|E|$. For $x\in \bb F_2^{E}$, define the \emph{local view} of $x$ at a vertex $v\in V$ to be $\res{x}{E(v)}$, which is the restriction of $x$ to $E(v)$, the set of edges incident to $v$. Then the codewords of $C$ are those $x\in \bb F_2^{E}$ such that $\res{x}{E(v)}\in C_0$ for every $v\in V$, where we choose some way of identifying every edge-neighborhood of a vertex with the bits of the local code $C_0$. If $H_0$ is the parity check matrix of $C_0$, then the parity check matrix of $C$ will have rows which are equal to a row of $H_0$ on an edge-neighborhood of a vertex and extended to be zero everywhere else. In the Tanner code construction, the code $C_0$ is often called the local, or base, code.

The dual of a classical linear code $C$, denoted $C^\perp$, is the subspace of all vectors orthogonal to the codewords of $C$; that is, 
\begin{align}
C^\perp=\{y\in \bb F_2^n: \forall x\in C,\ \langle x,y\rangle=0\}\, ,
\end{align}
where the inner product is taken modulo $2$. If we have two classical codes $C_A=\ker H_A\subseteq \bb F_2^n$ and $C_B=\ker H_B\subseteq \bb F_2^n$, we can consider their tensor code and dual tensor code. 

\begin{definition}[Tensor and Dual Tensor Codes]
The tensor code of $C_A$ and $C_B$ is the usual tensor product $C_A\otimes C_B\subseteq \bb F_2^n\otimes\bb F_2^n$. We can naturally interpret $\bb F_2^n\otimes\bb F_2^n$ as the set of binary $n\times n$ matrices, and in this view, $C_A\otimes C_B$ is identified with the set of matrices $X$ such that every column of $X$ is a codeword of $C_A$ and every row of $X$ is a codeword of $C_B$. 

The dual tensor code of $C_A$ and $C_B$ is $(C_A^\perp \otimes C_B^\perp)^\perp\subseteq \bb F_2^n\otimes\bb F_2^n$, which can equivalently be expressed as $(C_A^\perp \otimes C_B^\perp)^\perp=C_A\otimes \bb F_2^n + \bb F_2^n\otimes C_B$. Codewords of the dual tensor code are precisely the set of matrices $X$ such that $H_AXH_B^{\mathrm{T}}=0$.
\end{definition}

Note that if $C_A$ is a $[n_A,k_A,d_A]$ code and $C_B$ is a $[n_B,k_B,d_B]$ code, then their tensor code is a $[n_An_B, k_Ak_B, d_Ad_B]$ code. Their dual tensor code is a $[n_An_B, n_Ak_B+n_Bk_A-k_Ak_B, \min(d_A,d_B)]$ code. Moreover, we have $C_A\otimes C_B \subseteq (C_A^\perp \otimes C_B^\perp)^\perp$.

\subsection{Quantum CSS codes}

A quantum stabilizer code is a subspace $\mc C\subseteq \left(\bb C^2\right)^{\otimes n}$ that is the $+1$-eigenspace of an abelian subgroup $S$ of the $n$-qubit Pauli group. If $S$ can be generated by stabilizers that are products of $X$ operators and stabilizers that are products of $Z$ operators, we say that $\mc C$ is a CSS code. In this case, we can associate with $\mc C$ two classical codes $\eu C_X=\ker H_X$ and $\eu C_Z=\ker H_Z\subseteq \bb F_2^n$, where the rows of $H_X$ (resp. $H_Z$) specify the $X$ (resp. $Z$) type stabilizer generators. The property that $X$ and $Z$ generators commute translates to the condition $H_XH_Z^{\mathrm{T}}=0$, or equivalently $\eu C_Z^\perp\subseteq \eu C_X$. 

We can state the code parameters of a CSS code in terms of its underlying classical codes: if $\eu C_X$ (resp. $\eu C_Z$) has $k_X$ (resp. $k_Z$) encoded bits, then the number of encoded qubits is $k=k_X+k_Z-n$. The distance of the CSS code is given by $d=\min\{d_X,d_Z\}$, where
\begin{equation}
    d_X = \min_{x\in \eu C_Z\setminus \eu C_X^\perp} |x|\,, \quad d_Z = \min_{x\in \eu C_X\setminus \eu C_Z^\perp}|x|\, .
\end{equation}
We say that such a quantum code has parameters $[[n,k,d]]$. A family of quantum codes is called \emph{asymptotically good} (or simply \emph{good}) if the rate $\rho=k/n$ and the relative distance $\delta=d/n$ are bounded below by a non-zero constant. The code family is said to be low-density parity check (LDPC) if it can be defined with stabilizer generators that have at most constant weight, with each qubit being in the support of at most a constant number of generators. This is the case if each row and column of $H_X$ and $H_Z$ have at most constant weight.

\subsection{Left-Right Cayley complexes}
Let $G$ be a finite group with a symmetric generating set $A$, i.e. $A=A^{-1}$. The left\footnote{There is also the notion of a right Cayley graph $\Cay(G,A)$ where the generator set acts on the right, with edges $\{(g,ga):g\in G, a\in A\}$.} Cayley graph $\Cay(A,G)$ is the graph with vertex set $G$ and edge set $\{(g,ag): g\in G, a\in A\}$. Let $A$ and $B$ be two symmetric generating for $G$ of size $|A|=|B|=\Delta$. The generating sets $A$ and $B$ are said to satisfy the \emph{Total No-Conjugacy condition (TNC)}~\cite{Dinur21} if we have $ag \neq gb$ for all $a\in A$, $b\in B$, and $g\in G$.

Given a group $G$ and two symmetric generating sets $A$ and $B$ satisfying TNC, we define their double-covered left-right Cayley complex $\mathrm{Cay}_2(A,G,B)$ as the $2$-dimensional complex consisting of:
\begin{enumerate}
    \item Vertices $V = V_0\sqcup V_1 = G\times\{0\} \sqcup G\times \{1\}$.
    There are a total of $|V| = 2|G|$ vertices, with $|V_0| = |V_1| = |G|$.
    
    \item Edges $E = E_A \sqcup E_B$, where
    \begin{align}
    E_A = \{((g,0),(ag,1)) : g\in G, a\in A\}\, ,\quad\text{and}\quad E_B = \{((g,0),(gb,1)) : g\in G, b\in B\}\, .
    \end{align}
    Note that $A$-type edges are defined by a left-action of the generators, while that $B$-type edges are defined by a right-action of the generators. There are a total of $2\Delta|G|$ edges, with $|E_A| = |E_B| = \Delta|G|$. 
    
    \item Squares $Q$ defined by quadruplets of vertices:
    \begin{align}
        Q = \{\{(g,0),(ag,1),(gb,1),(agb,0)\} : a\in A, b\in B, g\in G\}\, .
    \end{align}
    There are a total of $|Q|=\Delta^2|G|/2$ squares.
\end{enumerate}

Note that the graph defined by $(V,E_A)$ is precisely the double cover of the left Cayley graph $\Cay(A,G)$, and the graph defined by $(V,E_B)$ is the double cover the right Cayley graph $\Cay(G,B)$. The full $1$-skeleton of $\Cay_2(A,G,B)$ is a bipartite graph $\mc G^\cup=(V,E)$.

By TNC, each square is guaranteed to have $4$ distinct vertices, so the graph $\mc G^\cup$ is a simple $2\Delta$-regular graph. There are $\Delta^2$ squares incident to a given vertex, and the set of faces incident to a given vertex can be naturally identified with the set $A\times B$. Figure~\ref{fig:local_view} illustrates the faces incident to a given vertex in the left-right Cayley complex.

\begin{figure}
    \centering
    \begin{tikzpicture}[scale=1]
\draw[thick] (0,0) circle(2);
\end{tikzpicture}
    \caption{The local view of a vertex $v$ and its identification with the set $A\times B$. Considering the ``book'' defined by the edge $b_1$ picks out a column, $A \times b_1$ (dashed). Specifying entries of $A\times B$ picks out specific faces (red, blue) of the local view, which can be regarded as entries of the corresponding matrix.}
    \label{fig:local_view}
\end{figure}
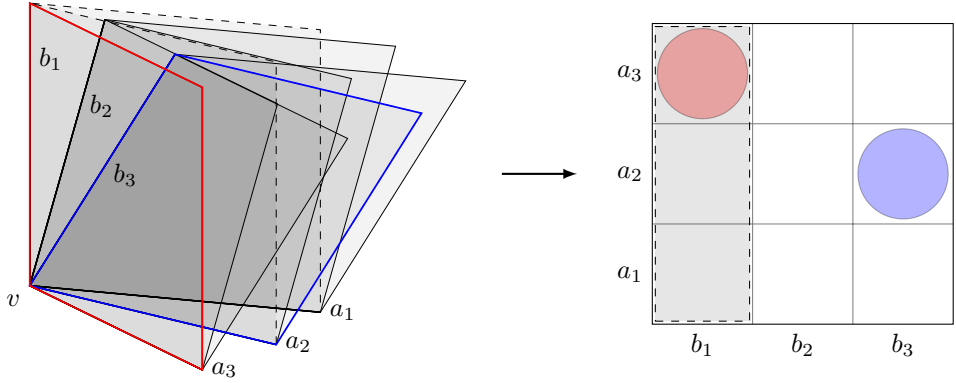

Based on the structure of the graph $\mc G^\cup$, each face $q\in Q$ can be naturally identified with its diagonal connecting its corners in $V_0$. Through this identification, we can define a graph $\mc G^\square_0$ capturing the incidence structure of faces in the complex. The graph $\mc G^\square_0 = (V_0,Q)$ is defined with vertex set $V_0 = G\times \{0\}$, where $q\in Q$ is present as an edge $(v,v')$ in $\mc G^\square_0$ if and only if $v$ and $v'$ appear as opposite $V_0$-corners of the square $q$. Likewise, each face $q\in Q$ can be identified with its diagonal connecting its corners in $V_1$. This similarly defines a graph $\mc G^\square_1 = (V_1,Q)$. Note that $\mc G_0^\square$ and $\mc G^\square_1$ are $\Delta^2$-regular multigraphs.

\subsection{Quantum Tanner codes construction}
We now describe the construction of quantum Tanner codes from~\cite{LZ22}. The construction is dependent on the choice of a double-covered left-right Cayley complex $\Cay_2(A,G,B)$ with generating sets of size $|A|=|B|=\Delta$ satisfying TNC. It is also dependent on fixed classical codes $C_A,C_B$ of blocklength $\Delta$, which define local codes $C_0=C_A\otimes C_B$ and $C_1=C_A^\perp\otimes C_B^\perp$. 

Given the data above, a quantum Tanner code $\mc C$ is then defined as the CSS code specified by the two classical Tanner codes $\eu C_Z = T(\mc G_0^\square,C_0^\perp)$ and $\eu C_X=T(\mc G_1^\square,C_1^\perp)$. More explicitly, qubits are placed on the squares of the left-right Cayley complex, and the $Z$ (resp. $X$) type stabilizer generators are codewords of the local code $C_A\otimes C_B$ (resp. $C_A^\perp \otimes C_B^\perp$) on the $\Delta^2$ squares incident to each vertex $v\in V_0$ (resp. $v\in V_1$). The incidence structure of the left-right Cayley complex ensures that the $X$ and $Z$ stabilizers commute (see Figure~\ref{fig:checks_commute}).

\begin{figure}
    \centering
    \begin{tikzpicture}[scale=1]
\draw[thick] (0,0) circle(2);
\end{tikzpicture}
    \caption{The restriction of checks to the faces incident to an edge $(v',v)\in \mc G^\cup$. The columns on the left indicate the various nontrivial restrictions of $Z$ stabilizers from the $v'$ local view, which are codewords of $C_A$. The columns on the right indicate the various $X$ stabilizer restrictions from the $v$ local view, which are codewords of $C_A^\perp$.}
    \label{fig:checks_commute}
\end{figure}
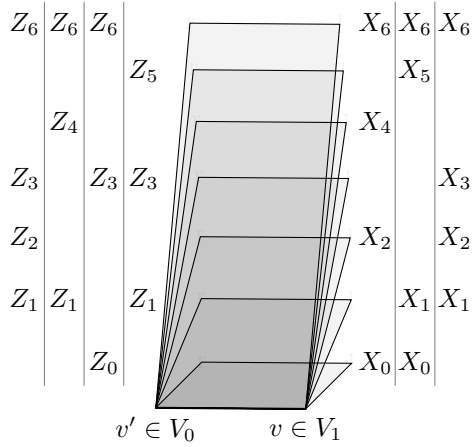

Note that $\mc C$ is a qLDPC code: each stabilizer generator acts on a subset of the local view $Q(v)$ of $\Delta^2$ qubits, and each qubit is acted on only by the stabilizers in the local views of its four corners. It is proven in~\cite{LZ22} that for certain choices of the left-right Cayley complex and local codes, this construction yields a good family of quantum codes:

\begin{theorem}[Theorem 16 of \cite{LZ22}]
\label{thm:QTCconstruction}
Fix $\varepsilon\in (0,1/2)$, $\rho\in (0,1/2)$, and $\delta\in (0,1/2)$ with $\delta<h^{-1}(\rho)$, where $h(x)=-x\log_2 x - (1-x)\log_2(1-x)$ is the binary entropy function. For some $\Delta$ sufficiently large, there exist classical codes $C_A,C_B$ of blocklength $\Delta$, rates $\rho$ and $1-\rho$ respectively, and relative distances at least $\delta$, as well as an infinite family of left-right Cayley complexes $\Cay_2(A,G,B)$ with $|G|\to\infty$ and symmetric generating sets $A,B$ of size $|A|=|B|=\Delta$ satisfying TNC, such that the quantum Tanner code defined above has parameters 
\[[[n=|Q|,k\ge (1-2\rho)^2n,d\ge \frac{\delta}{4\Delta^{3/2+\varepsilon}}n]]\, .\]
\end{theorem}

\subsection{Expanding Cayley complex and robust local codes}
In this subsection, we specify the technical properties of the Cayley complex and local codes that are used in the construction of good quantum Tanner codes described previously.

For a $D$-regular graph $\mc G=(V,E)$, the largest eigenvalue of its adjacency matrix is $\lambda_1=D$, and we let $\lambda(\mc G)=\lambda_2$ denote its second largest eigenvalue. The value of $\lambda(\mc G)$ is related to the expansion properties of the graph, as seen in the expander mixing lemma below. For subsets $S,T\subseteq V$, let $E(S,T)$ be the multiset of edges between $S$ and $T$, where edges in $S\cap T$ are counted twice. We have the following:

\begin{theorem}[Expander mixing lemma]
For a $D$-regular graph $\mc G=(V,E)$ and subsets $S,T\subseteq V$,
\begin{equation}
    |E(S,T)| \le \frac{D}{|V|}|S||T| + \lambda(\mc G)\sqrt{|S||T|}\, .
\end{equation}
\end{theorem}

The groups $G$ and generating sets $A,B$ in Theorem~\ref{thm:QTCconstruction} are chosen so that the resulting left-right Cayley complex has good expansion.

\begin{lemma}[Claim 6.7 of~\cite{Dinur21}]
\label{lem:PSL}
Let $q$ be an odd prime power and $G=\mathrm{PSL}_2(q^i)$. There exist two symmetric generating sets $A,B$ of size $|A|=|B|=\Delta=q+1$ and satisfying TNC such that the resulting Cayley graphs $\Cay(A,G), \Cay(G,B)$ are Ramanujan, i.e. have second largest eigenvalue $\lambda_2\le 2\sqrt \Delta$.
\end{lemma}

For $G,A,B$ as above, it can be shown~\cite{LZ22} that the relevant graphs in the quantum Tanner codes construction have the parameters specified in Table~\ref{tab:graph_parameters}.
\begin{table}[ht]
    \centering
    \begin{tabular}{c|c|c|c}
         Graph & Degree & Number of vertices & Second eigenvalue \\\hline
         $\mc G^\cup$ & $2\Delta$ & $2|G| = |V_0| + |V_1|$ & $\le 4\sqrt{\Delta}$\\
         $\mc G_0^\square$ & $\Delta^2$ & $|G| = |V_0|$ & $\le 4\Delta$\\
         $\mc G_1^\square$ & $\Delta^2$ & $|G| = |V_1|$ & $\le 4\Delta$
    \end{tabular}
    \caption{Graph parameters}
    \label{tab:graph_parameters}
\end{table}

The classical codes used in the construction of quantum Tanner codes are required to satisfy a robustness property of their dual tensor code, introduced in~\cite{LZ22}.

\begin{definition}[$w$-Robustness]
\label{def:robustness}
Let $C_A,C_B\subseteq \bb F_2^n$ be classical codes with distances $d_A$ and $d_B$ respectively. We say that the dual tensor code $C_{AB}=C_A\otimes \bb F_2^n + \bb F_2^n\otimes C_B$ is \emph{$w$-robust} if every codeword $X\in C_{AB}$ with $|X|\le w$ is supported on the union of at most $|X|/d_A$ non-zero columns and $|X|/d_B$ non-zero rows. That is, there exist rows $A'$ with $|A'|\ge n-|X|/d_B$ and columns $B'$ with $|B'|\ge n-|X|/d_A$ such that $\res{X}{A'\times B'}=0$.
\end{definition}

If the dual tensor code of $C_A$ and $C_B$ is $w$-robust, then their tensor code satisfies a property similar to robust testability defined in~\cite{BS06}.

\begin{proposition}[Proposition 6 of \cite{LZ22}]\label{prop:robustness}
Let $C_A,C_B\subseteq \bb F_2^n$ be classical codes with distances $d_A$ and $d_B$ respectively such that their dual tensor code is $w$-robust for $w\le d_Ad_B/2$. Then
\begin{equation}
    d(x,C_A\otimes C_B)\le \frac 3 2\left(d(x,C_A\otimes \bb F_2^n) + d(x,\bb F_2^n\otimes C_B)\right)
\end{equation}
whenever $d(x,C_A\otimes \bb F_2^n) + d(x,\bb F_2^n\otimes C_B)\le w$.
\end{proposition}

In Appendix~\ref{sec:robustness}, we prove Theorem~\ref{thm:dualtensorrobustness} below, which shows that for sufficiently large blocklengths, there exist dual tensor codes of sufficiently large robustness.

\begin{theorem}
\label{thm:dualtensorrobustness}
Fix constants $\varepsilon\in (0,1/28)$, $\rho\in (0,1/2)$, and $\delta\in (0,1/2)$ such that $\delta<h^{-1}(\rho)$, where $h(x)$ is the binary entropy function. For all sufficiently large $\Delta$, there exist classical codes $C_A,C_B$ of length $\Delta$ and rates $\rho_A= \rho$ and $\rho_B=1-\rho$ such that such that both the dual tensor code of $C_A$ and $C_B$ and the dual tensor code of $C_A^\perp$ and $C_B^\perp$ are $\Delta^{3/2+\varepsilon}$-robust and have distances at least $\delta\Delta$.
\end{theorem}

With these ingredients, we can describe the construction in Theorem~\ref{thm:QTCconstruction} in more detail. We first choose a prime power $q=\Delta-1$ sufficiently large such that we can use Theorem~\ref{thm:dualtensorrobustness} to find $C_A,C_B$ with robustness parameter $\Delta^{3/2-\varepsilon}$. Then the infinite family of left-right Cayley complexes is defined using $G=\mathrm{PSL}_2(q^i)$ for increasing values of $i$ and $A,B$ as in Lemma~\ref{lem:PSL}. Note that the sizes of the groups satisfy $|G|=\frac{1}{2}q^i(q^{2i}-1)\to \infty$.

We remark that in~\cite{LZ22}, a version of Theorem~\ref{thm:dualtensorrobustness} was shown for robustness parameter $\Delta^{3/2-\varepsilon}$, but in the proof of correctness of our decoder, a larger parameter $\Delta^{3/2+\varepsilon}$ is needed. Because the proof of Theorem~\ref{thm:QTCconstruction} given in~\cite{LZ22} is valid even for negative values of $\varepsilon$, the existence of dual tensor codes with higher robustness implies a larger distance of the code itself, $d\ge \frac{\delta}{4\Delta^{3/2-\varepsilon}}n$. At the same time, the larger robustness parameter eliminates the need for resistance to puncturing required in~\cite{LZ22}, thus simplifying the overall description of the quantum Tanner code.

\section{Decoding algorithm}\label{sec:decoder}

In this section, we give a description of our decoder for quantum Tanner codes. The quantum Tanner codes we consider are those described in the previous section with distance $d\ge \frac{\delta}{4\Delta^{3/2-\varepsilon}}n$, constructed using classical dual tensor codes of robustness $\Delta^{3/2+\varepsilon}$ as the local codes. In the decoding problem, an unknown (Pauli) error is applied to the code. We may only extract the syndrome of the error by measuring stabilizers, and based on the syndrome, apply corrections. We succeed in decoding if the correction we applied is equal to the error, up to a stabilizer (which has no effect on the codespace). Because quantum Tanner codes are CSS codes, it suffices to consider $X$ and $Z$ errors separately. If we have an algorithm to correct for errors that are purely a product of $X$ operators and another one for a product of $Z$ operators, a general error will be corrected after running both algorithms. Furthermore, since the code is symmetric between $X$ and $Z$, we just consider the problem of correcting $Z$ errors.

\begin{definition}[Decoding Problem]
Let $e\in \bb F_2^{Q}$ be a $Z$ error. Given the syndrome $\sigma=H_Xe$ as input, the task of the decoding problem is to output a correction $f\in \bb F_2^{Q}$ such that $e-f\in \eu C_Z^\perp$.
\end{definition}

Our decoder is similar in flavor to the small-set-flip decoder used on certain hypergraph product codes~\cite{LTZ15}. Small-set-flip is an iterative decoder, where in each step the decoder tries to decrease the syndrome weight by looking for corrections within the support of a $Z$ generator. If the initial error weight is less than the code distance, then such a correction can always be found, and this implies that the decoder can successfully errors of weight less than a constant fraction of the code distance~\cite{LTZ15}.

In our case, the syndrome weight is not a very well-defined concept due to the presence of the local codes. Because the $X$ stabilizers are generated by local tensor codes $C_1=C_A^\perp\otimes C_B^\perp$, defining the Hamming weight of the syndrome involves choosing a basis for $C_1$. Unfortunately, there is no canonical choice of basis, and different choices will give different Hamming weights of a given error. We address this issue by introducing the concept of a potential function. Recall that an element $x\in\bb F_2^{Q}$ is a codeword of $\eu C_X=T(\mc G_1^\square,C_1^\perp)$ if and only if every local view $\res{x}{Q(v)}$, $v\in V_1$ is a codeword of $C_1^\perp$. We define the potential by the distance of the local view to the codespace, which can be inferred from the syndrome. More formally, we have the following definition:

\begin{definition}[Local and Global Potential Functions]
Let $e\in \bb F_2^{Q}$ be an error. Define the \emph{local potential} at a vertex $v\in V_1$ by the Hamming distance
\begin{equation}
    U_v(e) = d\left(\res{e}{Q(v)},C_1^\perp\right)\, .
\end{equation}
The \emph{global potential} is defined as
\begin{equation}
    U(e) = \sum_{v\in V_1} U_v(e)\, .
\end{equation}
\end{definition}

The local potential is the minimum weight of a correction that is needed to take the local view of the error (or corrupted codeword) back into the local codespace $C_1^\perp$. Thus, it is a quantity that can be computed just from the syndrome. We will abuse notation and also write $U_v(\sigma)=U_v(e)$ and $U(\sigma)=U(e)$. Note that in absence of a local code, in other words a local code where the codewords are the vectors of even Hamming weight, the local potential is simply either 0 or 1 depending on if the constraint is satisfied, so it coincides with the Hamming weight of the syndrome.

Our decoding algorithm (Algorithm~\ref{alg:decoder}) runs by looking for bits to flip in local views that will decrease the global potential.

\begin{algorithm}
	\caption{Decoder for quantum Tanner codes}
	\label{alg:decoder}
	\textbf{Input:} A syndrome $\sigma=H_Xe\in \bb F_2^{|V_1|\dim C_1}$ of an error $e\in\bb F_2^{Q}$.\\
	\textbf{Output:} A correction $f\in \bb F_2^{Q}$ for $e$.
	\begin{algorithmic}
    \State $f\gets 0$
    \State $U\gets U(\sigma)$
    \While {$U>0$} 
      \State Look for a vector $z\in\bb F_2^{Q}$ supported on a local view $Q(v)$, $v\in V_0\cup V_1$ such that $U(\sigma+H_Xz)<U$
      \State $f\gets f+z$
      \State $\sigma\gets \sigma+H_Xz$
      \State $U\gets U(\sigma)$
    \EndWhile\\
    \Return $f$
    \end{algorithmic}
\end{algorithm}

We will show that Algorithm~\ref{alg:decoder} succeeds in the decoding problem if the initial error has weight at most a constant fraction of code distance; that is, it can correct all errors up to some linear weight. The main difficulty of the proof is in showing that there always exists a vector $z$ that decreases the global potential when flipped. This is captured in the following theorem, which we prove in the next section.

\begin{theorem}
\label{thm:correction_exists}
Let $e\in\bb F_2^{Q}$ be an error of weight $|e| \le \delta n/6\Delta^{3/2-\varepsilon}$  with syndrome $\sigma=H_Xe$. Then there exists $v \in V_0\cup V_1$ and some $z\in\bb F_2^{Q}$ supported on the local view $Q(v)$, such that $U(\sigma+H_Xz)<U(\sigma)$.
\end{theorem}

From this property, we can show that the algorithm will output a valid correction. We do this by proving a statement that applies to a more general class of small-set-flip type decoders based on a potential function. The proof follows the same idea as that of Lemma 10 in~\cite{LTZ15}.

\begin{lemma}
\label{lem:decoder_from_potential}
Let $\alpha<1, s, c$ be constants. Let $\mc C$ be an $[[n,k,d]]$ quantum CSS code defined by the classical codes $\eu C_X,\eu C_Z\subseteq\bb F_2^n$. Let $U: \bb F_2^n \to \mathbb Z_{\ge 0}$ be a (global) potential function that is constant on cosets of $\eu C_X$, satisfies $U(e)=0$ if and only if $e\in \eu C_X$, and $U(e)\le s|e|$ for all $e\in \bb F_2^n$. Suppose we have an iterative decoder that, given the syndrome of a non-zero $Z$ error of weight less than $\alpha d$, can decrease the potential by applying an $X$ operator of weight at most $c$. Then the decoder can successfully correct errors of weight less than $\alpha d/(1+sc)$.
\end{lemma}

\begin{proof}
Let $x'=x+e\in \bb F_2^n$ be a corrupted codeword with $x\in\eu C_X$ and error $e$ of weight $|e|<\frac{\alpha d}{1+sc}$. The decoder outputs a sequence of corrections $0=f_0, f_1, f_2,\dots$ such that the resulting errors $e_i=e+f_i$ satisfy $|e_{i+1}-e_i|\le c$ and $U(e_i)-U(e_{i+1})\ge 1$ for all $i$. Suppose we have decoded up to step $j$. Then
\begin{align}
    |e_j| &\le |e_0| + |e_1-e_0| + \dots + |e_j-e_{j-1}|\\
    &\le |e| + c + \dots + c\\
    &\le |e| + c(U(e_0)-U(e_1)) + \dots + c(U(e_{j-1})-U(e_j))\\
    &= |e| + c(U(e_0)-U(e_j))\\
    &\le (1+sc)|e|\\
    &< \alpha d\, .
\end{align}
So either $U(e_j)=0$, or the decoder can find the next correction $f_{j+1}$ to produce $e_{j+1}$. Eventually, the decoder will output $e_J$ such that $U(e_J)=0$. In other words, $e_J\in\eu C_X$. But since $|e_J|<\alpha d<d$, it must be in $\eu C_Z^\perp$, and we have decoded to the correct codeword.
\end{proof}

We can now state our main theorems.

\begin{theorem}\label{thm:adversarial_errors}
Fix $\varepsilon\in (0,1/28)$, $\rho\in (0,1/2)$, and $\delta\in (0,1/2)$ with $\delta<h^{-1}(1-\rho)$, where $h(x)$ is the binary entropy function. For some $\Delta$ sufficiently large, there is an infinite family of quantum Tanner codes with parameters
\[[[n,k\ge (1-2\rho)^2n,d\ge \frac{\delta}{4\Delta^{3/2-\varepsilon}}n]]\]
with $n\to\infty$, such that for each $n$, Algorithm~\ref{alg:decoder} can correct all errors of weight
\begin{equation}
    |e| \le \frac{\delta n}{6\Delta^{3/2-\varepsilon}(1+2\Delta^2)}\, .
\end{equation}
\end{theorem}

\begin{proof}
The infinite family of quantum Tanner codes is as described in Section~\ref{sec:quantum_tanner} (with distance parameter from the improved robustness of the classical local codes). To prove the decodable distance, consider the parameters in Lemma~\ref{lem:decoder_from_potential}. Every bit in an error can at most increase the local potentials of the two incident $V_1$ vertices by one each. This implies the bound $U(e)\le 2|e|$, so we can take $s=2$. Since at each step, the algorithm flips sets within a local view, we set $c=\Delta^2$. From Theorem~\ref{thm:correction_exists}, the decoder can reduce the global potential when the error has weight up to $\alpha d=\delta n/6\Delta^{3/2-\varepsilon}$. The theorem then follows from Lemma~\ref{lem:decoder_from_potential}.
\end{proof}

\begin{theorem}
\label{thm:linear_time}
Algorithm~\ref{alg:decoder} runs in time $O(n)$.
\end{theorem}

\begin{proof}
To compute the global potential $U$, we must compute $O(n)$ local potentials. Each local potential is a function of the constant-sized local view and can be computed in $O(1)$ time by enumerating vectors supported in the local view. At the same time, we can store the best candidate correction for the local view. Thus, the initialization runs in time $O(n)$.

In each iteration, we apply corrections in a constant-sized region, so only a constant number of local views and candidate corrections need to be updated for the syndrome and local potentials by the LDPC property. Each iteration of the algorithm runs in a constant amount of time, and there can be at most $O(n)$ iterations. Hence, the total runtime of Algorithm~\ref{alg:decoder} is $O(n)$.
\end{proof}

The correctness of the decoding algorithm implies a form of soundness for the quantum code. This notion is a related to local testability but weaker because it only applies to errors of sufficiently small weight.

\begin{corollary}[Soundness]\label{cor:soundness}
If $e$ is an error that is correctable using Algorithm~\ref{alg:decoder}, then $U(e)\ge \Delta^{-2}d(e,\eu C_Z^\perp)$.
\end{corollary}

\begin{proof}
Using Algorithm~\ref{alg:decoder}, $e$ can be corrected to a codeword of $\eu C_Z^\perp$ in at most $U(e)$ steps. In each step, at most $\Delta^2$ bits are flipped. Therefore, we have $d(e,\eu C_Z^\perp)\le \Delta^2 U(e)$.
\end{proof}

\begin{corollary}[Threshold]\label{cor:threshold}
Let $e \in \bb F_2^n$ be a random error with each entry independently and identically distributed such that $e_i = 1$ with probability $p$ and $e_i = 0$ with probability $1-p$. Under this model, the probability that Algorithm~\ref{alg:decoder} fails to return a correction $f$ such that $e+f \in \eu C_Z^\perp$ is $O(e^{-a n})$, with $a>0$, so long as $p<p^*$, where
\begin{align}
p^* \equiv \frac{\delta}{6 \Delta^{3/2-\varepsilon} ( 1 + 2 \Delta^2)}
\end{align}
is a lower bound for the accuracy threshold under independent bit and phase flip noise.
\end{corollary}

\begin{proof}
By Theorem~\ref{thm:adversarial_errors}, the decoder is guaranteed to succeed as long as $|e| \le n p^*$. The Hamming weight of $e$ is distributed as a Binomial random variable which concentrates around the mean $n p$. For $p^*>p$, we can use Hoeffding's inequality to bound the probability that $|e| > n p^*$ as
\begin{align}
\Pr \left( |e| > n p^* \right) < e^{-2 n(p^*-p)^2}\, ,
\end{align}
which completes the proof.
\end{proof}

\section{Proof of Theorem~\ref{thm:correction_exists}}\label{sec:proof}

Before beginning the proof of Theorem~\ref{thm:correction_exists}, we first elaborate on some conventions and notation. In the remainder of the paper we will adopt the convention that a vector $x \in \mathbb{F}_2^Q$ is treated equivalently as the subset of $Q$ indicated by the vector. This allows us to write expressions such as $x \cup y \in \mathbb{F}_2^Q$ to denote the vector defined by the union of $x, y\subseteq Q$. 

We will often need to consider the restriction of a vector $x\in \mathbb{F}_2^Q$ to the set of faces $Q(v)$ incident to some vertex $v\in V$. This is called the \emph{local view} of $x$ at $v$. In a convenient abuse of notation, we will equivalently consider local views as elements of $\mathbb{F}_2^{Q(v)}$, or as elements of $\mathbb{F}_2^Q$ with support on $Q(v)$. For simplicity of notation, we write local views at $v\in V$ with a subscript $v$, for example $x_v=\res{x}{Q(v)}$. 

By the TNC condition, $Q(v)$ is in bijection with $A\times B$ so that each local view naturally defines a $\Delta\times\Delta$ matrix, i.e., $x_v \in \mathbb{F}_2^{\Delta\times\Delta}$. We will label the faces of $Q(v)$ by pairs of vertices $v_1,v_2$, where $v_1$ is connected to $v$ by an edge in $A$, and $v_2$ to $v$ by an edge in $B$. In this case, we denote the unique face defined by these vertices by $[v_1,v_2]\in Q(v)$ and we say that $v_1$ is a row vertex for $v$, and that $v_2$ is a column vertex. We will use the notation $x_v[v_1,v_2]$ to denote the entry of $x_v$ specified by the face $[v_1,v_2]$. Likewise, we will adopt the notation $x_v[v_1,\cdot]$ to denote the row of $x_v$ indexed by the row vertex $v_1$, and similarly $x_v[\cdot,v_2]$ to denote the column of $x_v$ indexed by $v_2$. Given neighboring vertices $v\in V_0$ and $v'\in V_1$, the shared row (resp. column) of the local views $x_v$ and $x_{v'}$ can be equivalently denoted by either $x_v[v',\cdot]$ or $x_{v'}[v,\cdot]$ (resp. $x_v[\cdot,v']$ or $x_{v'}[\cdot,v]$). 

Let us now define the notion of a local minimum weight correction and other associated objects.

\begin{definition}
Let $e \in \mathbb{F}_2^Q$ be a $Z$ error. For each vertex $v\in V_1$, we define $c_v(e)$ as a closest codeword in $C_1^\perp$ to the local view $e_v$. If there are multiple closest codewords, then we may fix an arbitrary one.

For each vertex $v\in V_1$, let $R^+_{v}(e)=e_v-c_v(e)\subseteq Q(v)$. Then we call $R^+_v(e)$ the \emph{local minimum weight correction} at the vertex $v$. We will denote the collection of all local minimum weight corrections by $\mc R(e) = \{R^+_v(e)\}_{v\in V_1}$. We will also define the total correction
\begin{align}
    R(e)=\bigcup \mathcal{R}(e) = \bigcup_{v\in V_1}R_{v}^+(e)\, .
\end{align}
\end{definition}

Note that the local potential at $v$ is given by
\begin{align}
U_v(e) = d\left(e_v,C_1^\perp\right) = \left|e_v-c_v(e)\right| = |R^+_v(e)|\, ,
\end{align}
and our goal is to reduce the global potential $U(e)=\sum_{v\in V_1} U_v(e)$ at every step of the decoding. When the error $e$ is understood, we will often simply write $c_v$, $R^+_v$, and $R$ for short.

We can now proceed with the proof of Theorem~\ref{thm:correction_exists}, which we split into three cases: 
\begin{enumerate}
\item In the first case, we consider whether flipping single qubits can decrease the total potential. If this is not the case, it will introduce extra structure in the set $R$. 

\item In the second case, we ask if $R$ has high overlap with a codeword of $C_1^\perp$ in a $V_1$ local view. If so, it will allow us to flip a set of qubits that together can decrease the total potential. 

\item The third and most complicated case is the one complementary to the first two, where no single qubit flip can decrease the total potential, and where $R$ has low overlap with all local codewords. The intuition here is that $R$ cannot be a very large set, so every $V_1$ local view of the error is close to the local code. Because the error ``looks like'' a codeword, we are able to apply reasoning similar to the local minimality argument in the proof of the distance of the code. In essence, the expansion of the graph allows us to find a special $V_0$ vertex whose local view contains a flip set to decrease the total potential.
\end{enumerate}
\subsection{Proof of Cases 1 and 2}

In this subsection, we prove Theorem~\ref{thm:correction_exists} for the first two cases listed above. The terminology and definitions established in this subsection will also be crucial to the proof of case 3. To consider the first case, we define the concept of a metastable configuration.

\begin{definition}
Let $e \in \mathbb{F}_2^Q$ be an error. We say that $e$ is \emph{metastable} if flipping any one qubit $q\in Q$ does not decrease the global potential. We also say that $\mc R(e)$ and $R(e)$ are metastable if they are obtained from a metastable error $e$. Note that while we only define and use metastability for an error $e$ and its configuration of local minimum weight corrections, the property of metastability is really a property intrinsic to the underlying syndrome $\sigma$.
\end{definition}

Note that case 1 pertains precisely to the case when the error $e$ is not metastable. If $e$ is not metastable then there exists some $q\in Q$ which decreases the global potential and Theorem~\ref{thm:correction_exists} follows. Therefore, in the remainder of this section we consider the case that $e$ (and hence $R$) is metastable.

\begin{definition}
Let $e\in\mathbb{F}_2^Q$ be an error, and let $\mathcal{R} = \{R^+_v(e)\}_{v\in V_1}$ be a set of local minimum weight corrections for $e$. We say that $\mc R$ is \emph{disjoint} if $R^+_v(e) \cap R^+_{v'}(e) = \emptyset$ for all $v\neq v'$.
\end{definition}

When $\mathcal{R}$ is a disjoint set of corrections we can think of it as a directed subgraph of $\mathcal G_1^{\square}$ by viewing each $R^+_v$ as the set of outgoing edges from $v$ (see Figure~\ref{fig:g_square1}). The local view $R_v$ is then the set of all edges, incoming or outgoing, incident to $v$ in this directed graph. Note that in this case, the set $\mathcal{R}$ completely defines the underlying directed graph. Conversely, given the directed subgraph, we may uniquely recover $\mathcal{R}$ by taking $R^+_v(e)$ as the set of outgoing edges at each vertex. Therefore we will identify a disjoint $\mathcal{R}$ with the directed subgraph it defines in the following. We can likewise identify the set of total corrections $R$ with the undirected graph underlying $\mc R$.

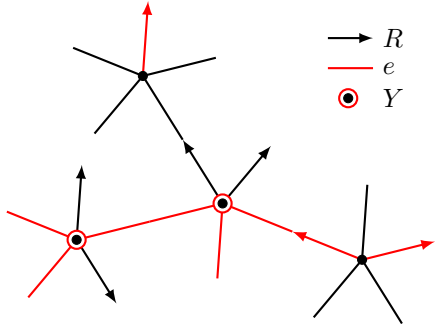
\begin{figure}
    \centering
    \begin{tikzpicture}[scale=1]
\draw[thick] (0,0) circle(2);
\end{tikzpicture}
    \caption{Subsets of $\mc G_1^\square$ indicating $e$, $R$, and elements of $Y$. In the diagram $C_1^\perp$ is the repetition code (codewords $00000$ and $11111$). Note that the red edges without an arrow are in $y$, the red edges with an arrow are in $R\cap e$, and the undecorated black edges are just the remaining edges in $\mc G^\square_1$.}
    \label{fig:g_square1}
\end{figure}

Note that $\mathcal{R}$ will always be disjoint when $e$ is a metastable error (otherwise flipping a shared qubit will lower the global potential by $2$). For a metastable error, flipping a qubit $q=(v,v')\in R_v^+$, which is a directed edge from $v$ to $v'$, decreases $U_v$ by one and increases $U_{v'}$ by one. We first prove a lemma which shows that metastable errors are somewhat rigid under additional bit-flips.

\begin{lemma}[$R$-flipping]\label{lem:R_flip}
Let $\mc R(e)$ be a directed subgraph of $\mathcal G_1^\square$ corresponding to a set of local minimum weight corrections for a metastable error $e$. Suppose furthermore that for some subset $\hat{R}\subseteq R(e)$, flipping all qubits of $\hat{R}$ does not decrease the global potential. Consider the error $e+\hat{R}$. Then a valid configuration $\mc R(e+\hat{R})$ of locally minimum weight corrections for $e+\hat{R}$ is obtained from $\mc R(e)$ by reversing the directions of all edges in $\hat{R}$. Moreover, the nearest codewords $c_v$ at each vertex remains unchanged, i.e., 
\begin{align}
c_v(e) = e_v + R^+_v(e) = (e+\hat{R})_v + R^+_v(e+\hat{R}) = c_v(e+\hat{R})\, .
\end{align}
\end{lemma}

\begin{proof}
Consider any $v\in V_1$. By definition, each $R_v^+(e)$ is a minimum weight correction to the local code at $v$, so $c_v(e)=e_v+R_v^+(e)$ and $U_v(e)=|R_v^+(e)|$. Now suppose we flip all qubits in $\hat{R}$. In the local view of $v$, we have
\begin{align}
    c_v(e) &= e_v + \hat{R}\cap Q(v) + R_v^+(e) + \hat{R}\cap Q(v)\\
    &= (e+\hat{R})_v + R_v^+(e) + \hat{R}\cap R_v^+(e) + \hat{R}\cap R^-_v(e)\, ,\label{eq:Rflip1}
\end{align}
where we define $R^-_v(e) = R_v(e)\backslash R_v^+(e)$. Note that $R^-_v(e)$ can be thought of as the set of incoming edges at $v$ in the directed graph defined by $\mc R(e)$. Therefore, we can bound the weight of the new minimal weight correction for vertex $v$ by
\begin{align}
    U_v(e+\hat{R}) &\le |R_v^+(e) + \hat{R}\cap R_v^+(e) + \hat{R}\cap R^-_v(e)|\label{eq:Uv_inequality}\\
    &= |R_v^+(e) + \hat{R}\cap R_v^+(e)| + |\hat{R}\cap R^-_v(e)|\\
    &= U_v(e) - |\hat{R}\cap R_v^+(e)| + |\hat{R}\cap R^-_v(e)|\, , 
\end{align}
where the first line follows from equation~\eqref{eq:Rflip1} and the second from the disjointness of the sets $R_v^+(e)$ and $R^-_v(e)$. Note that if equality holds in equation~\eqref{eq:Uv_inequality}, then a valid minimum weight correction for $(e+\hat{R})_v$ is given by 
\begin{align}
R_v^+(e+\hat{R}) = R_v^+(e) + \hat{R}\cap R_v^+(e) + \hat{R}\cap R^-_v(e)\, .
\end{align}
The set $R_v^+(e+\hat{R})$ above is obtained from $R_v^+(e)$ by removing all outgoing edges in $\hat{R}$ and changing all incoming edges in $\hat{R}$ to outgoing edges. Also note that in this case the nearest codeword remains $c_v(e)$.

Summing inequality~\eqref{eq:Uv_inequality} for all $v\in V_1$ gives a bound on the global potential as
\begin{align}
    U(e+\hat{R}) &\le \sum_{v\in V_1}U_v(e) - \sum_{v\in V_1} |\hat{R}\cap R_v^+(e)| + \sum_{v\in V_1} |\hat{R}\cap R_v^-(e)|\\
    &= U(e) - |\hat{R}\cap R(e)| + |\hat{R}\cap R(e)|\\
    &= U(e)\, ,
\end{align}
where in the second line we've used the fact that $R(e) = \bigsqcup_{v\in V_1} R_v^+(e) = \bigsqcup_{v\in V_1}R^-_v(e)$ by metastability. By the assumption of the lemma, $U(e+\hat{R})\ge U(e)$. This means inequality~\eqref{eq:Uv_inequality} must hold with equality for all $v\in V_1$. Hence, we have proven that $\mc R(e+\hat{R})$ can be taken as $\mc R(e)$, but with the directions of edges in $\hat{R}$ reversed.
\end{proof}

\begin{remark}
In the scenario of the $R$-flipping lemma, while the error $e+\hat{R}$ may not be metastable itself, the set $\mc R(e+\hat{R})$ as defined as in the lemma is still disjoint. This new set is a valid correction in the sense that each $R_v^+(e+\hat{R})$ gives a minimum weight correction to the local code --- correcting the error $(e+\hat{R})_v$ to $c_v(e+\hat{R}) = c_v(e)$ --- at every $v\in V_1$. Note that the set of total corrections remains invariant in this case, i.e., $R(e) = R(e+\hat{R})$.
\end{remark}

In the second case, we assume that $R$ has high overlap with a codeword of $C_1^\perp$. We formalize this property below.

\begin{definition}[Low Overlap]\label{def:low_overlap}
    The set $R$ is said to have the \emph{low-overlap} property at $v\in V_1$ if for all codewords $c\in C^\perp_1$, we have $|R_v \cap c| \le |c|/2$. We will say that the set $R$ has the low-overlap property if it has the low-overlap property at every $v\in V_1$.
\end{definition}

Before formally proving case $2$, let us first provide some rough intuition. When the low-overlap property is not satisfied, there exists some codeword $c\in C_1^\perp$ at some vertex $v\in V_1$ which has large agreement with $R_v$. Using the $R$-flipping Lemma~\ref{lem:R_flip}, we may assume without loss of generality that $R^+_v = 0$. Now imagine flipping the set $R_v\cap c$. Since $R^+_v = 0$, every edge in $R_v$ belongs to a local correction neighboring $v$. Flipping $R_v \cap c$ will therefore lower the local potential at each of these neighbors by $1$. It will also raise the local potential at $v$, which was zero before. However, since $R_v$ has large overlap with $c$ it is actually more efficient to apply the correction $c \backslash R_v$ instead of $R_v\cap c$. In this case, the local error is pushed out of the neighborhood of its original nearest codeword $c_v(e)$ and into the neighborhood of $c_v(e) + c$ instead. The local potential at $v$ is therefore raised by an amount less than $R_v\cap c$, which results in an overall lowering of the global potential. Figure~\ref{fig:edge_flip} illustrates the proof technique.

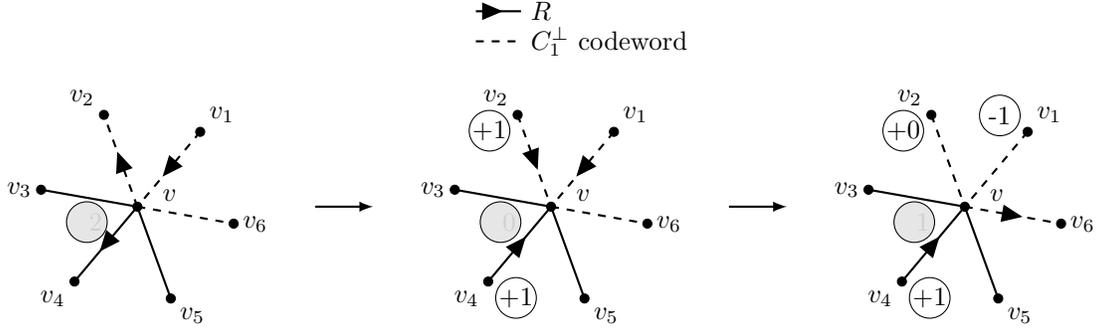
\begin{figure}
    \centering
    \begin{tikzpicture}[scale=1]
\draw[thick] (0,0) circle(2);
\end{tikzpicture}
    \caption{Flipping bits to decrease the global potential in case 2. The changes in local potentials after flipping the edges $(v, v_2)$, $(v, v_4)$ (left to center) and then flipping $(v,v_1)$, $(v,v_2)$, $(v,v_6)$ (center to right) in the graph $\mc G_1^\square$ are shown. The local potentials at $v$ are indicated within the shaded circles. Potential differences relative to the first configuration are indicated for the neighboring vertices. }
    \label{fig:edge_flip}
\end{figure}

\begin{lemma}
Let $R$ be metastable. If $R$ does not have the low-overlap property, then there exists $v \in V_1$ and a subset $f \subseteq Q(v)$ such that flipping the qubits of $f$ decreases the total potential.
\end{lemma}

\begin{proof}
Suppose that $R$ is metastable and does not have the low-overlap property. Then there exists some $v\in V_1$ and some $c\in C^\perp_1\setminus \{0\}$ such that $|R_v(e) \cap c| > |c|/2$. Let $e' = e+R_v^+(e)$. If $U(e') < U(e)$ then we are done. Otherwise $U(e')=U(e)$, and by the $R$-flipping Lemma~\ref{lem:R_flip}, we may take $R(e') = R(e)$ with $R_v^+(e') = 0$.

Consider now flipping the additional set of qubits $f'= R_v(e')\cap c$ to obtain the error $e'' = e'+f'$. For each $q = (v',v) \in f'$, we have $q \in R_{v'}^+(e')$, so that $|R^+_{v'}(e'')| = |R^+_{v'}(e')|-1$. This is the new value of the local potential at $v'$. Since we had $U_v(e') \equiv |R_v^+(e')| = 0$, the change in the global potential is given by $U(e'')-U(e') = U_v(e'')-|f'|$.

Since $e'_v \in C_1^\perp$, a valid correction for $e''_v$ is given by $f'+c$, where $c$ is the high-overlap codeword from earlier. This correction has weight $|f'+c| = |R_v(e)\cap c + c| < |c|/2 < |R_v(e)\cap c| = |f'|$. Therefore $U_v(e'')-|f'| < 0$, and we have $U(e'') < U(e') = U(e)$. Our desired flip-set is therefore $f=R_v^+(e)+R_v(e)\cap c$. 
\end{proof}

\subsection{Proof of Case 3}

The preceding subsection proves Theorem~\ref{thm:correction_exists} in the cases when $R$ is not metastable, or when $R$ is metastable but does not have the low-overlap property. In what follows, we consider the remaining case where $R$ is both metastable and has the low-overlap property. We summarize our key list of assumptions for this case below for convenience.

\begin{assumption}\label{key_assumptions}
Let $e \in \mathbb{F}_2^Q$ be a $Z$ error of weight $|e| \le \delta n/6\Delta^{3/2-\varepsilon}$. We assume that $e$ is a reduced error, i.e., it is the minimum weight element of the coset $e + \eu{C}_Z^\perp$. We assume that $e$ is a metastable error, and that its set of local minimum weight corrections $R(e)$ satisfies the low-overlap property~\ref{def:low_overlap}. Finally, we also require that the underlying quantum Tanner code be defined using dual tensor codes of sufficiently large robustness, i.e., with robustness parameter $\Delta^{3/2+\varepsilon'}$ for some $\varepsilon' > 0$. Throughout the rest of the proof, we fix any $\varepsilon < \varepsilon'$.
\end{assumption}

The proof of case $3$ proceeds in two general steps. In the first step, we show using the expansion of the underlying graphs that, given an error $e$ of sufficiently low weight, there always exists a special vertex $v_0 \in V_0$ with the property that $v_0$ ``sees'' many non-trivial codewords of $C_A$ and $C_B$ amongst its shared local views with the minimum weight corrections on neighboring vertices.

The second step of the proof proceeds to analyze the local view at the vertex $v_0$ described above. We show that due to the pattern of its many shared codewords, it is either the case that $R_{v_0} \subset Q(v_0)$ is sufficiently large to contain a flip-set which reduces the potential, or else it is small enough that $e_{v_0}$ has many columns and rows which are close to non-trivial codewords of $C_A$ and $C_B$. In the latter case, the robustness of the underlying dual tensor code then implies that $e_{v_0}$ must have sufficient overlap with a $Z$-stabilizer that the addition of this stabilizer will reduce the weight of $e$. Since we began without loss of generality with a reduced error $e$, this leads to a contradiction.

\subsubsection{Existence of $v_0 \in V_0$}

In the first part of the analysis of the third case, we proceed in a manner parallel to the proof of Theorem 1 in~\cite{LZ22}. The goal is to show that for an error $e$ with weight $|e| \le \delta n/6\Delta^{3/2-\varepsilon}$, there always exists a vertex $v_0 \in V_0$ whose local view contains many columns and rows which are close to non-trivial codewords of $C_A$ and $C_B$. Aside from some differences in definitions, the proofs and results of this subsection are equivalent to their counterparts in~\cite{LZ22}.

Since our goal is to find a vertex $v_0 \in V_0$ whose local view has many rows and columns close non-trivial codewords, we first parametrize the vertices of $V_1$ with non-trivial nearest codewords. This is captured by the set $Y$ below.

\begin{definition}
    Let $e \in \mathbb{F}_2^Q$ be an error and let $\mathcal{R} = \{R^+_v(e)\}_{v\in V_1}$ be a set of local minimum weight corrections. We define the set of non-trivially corrected vertices $Y\subseteq V_1$ as
    \begin{align}
        Y = \{v \in V_1 \mid R_v^+ \neq e_v\}\, .
    \end{align}
    That is, a vertex $v$ is in $Y$ if and only if the result of applying the locally minimum weight correction at $v$ results in a non-trivial codeword i.e. $c_v=e_v + R_v^+\neq 0$.
\end{definition}

To work with the vertex set $Y$, it will also be convenient to define an edgewise version of the condition $R^+_v\neq e_v$. To that end, we introduce the set $y$ of ``residual errors''. Given an error $e\in\mathbb{F}^Q_2$, the elements of $y$ are all of the elements of $e$ which have no overlap with the set of minimum weight corrections $R(e)$ (see Figure~\ref{fig:g_square1}).

\begin{definition}
    Let $e \in \mathbb{F}_2^Q$ be an error and let $\mathcal{R} = \{R^+_v(e)\}_{v\in V_1}$ be a set of local minimum weight corrections. The set of ``residual'' errors is defined by $y = e\backslash R \in \mathbb{F}_2^Q$, i.e., $y$ labels the set of errors which are not in any of the local minimum weight corrections.
\end{definition}

The edges of $\mc G^\square_1$ indexed by $y$ define a subgraph of $\mc G^\square_1$ which we will call $G_{1,y}^\square$. This subgraph is closely related to the set $Y$. It is straightforward to see that every vertex of $\mc G_{1,y}^\square$ must belong to $Y$. Conversely, the low-overlap property implies that each vertex of $Y$ must be incident to many edges in $\mc G_{1,y}^\square$. This means that $Y$ is precisely the vertex set of $\mc G_{1,y}^\square$ and moreover $\mc G_{1,y}^\square$ must have large minimum degree. This discussion is formalized below by Lemmas~\ref{lem:y_subgraph} and~\ref{lem:min_deg}.

\begin{lemma}\label{lem:y_subgraph}
Let $(v,v') \in y$ be an edge in $\mathcal{G}_{1,y}^\square$. Then both $v$ and $v'$ are elements of $Y$.
\end{lemma}
\begin{proof}
By definition, the edge $(v,v')\in y$ is an element of $e$ but not of $R$. Therefore $(v,v')$ is an element of $e_v$ (and likewise, of $e_{v'}$) but not an element of $R_v^+$ (and likewise, $R_{v'}^+)$. It follows that $e_v \neq R_v^+$ and $e_{v'} \neq R_{v'}^+$.
\end{proof}

\begin{lemma}\label{lem:min_deg}
    Every vertex $v\in Y$ is incident to at least $\delta\Delta/2$ edges in $y$. In particular, the subgraph $\mathcal{G}_{1,y}^\square$ has vertex set equal to $Y$ and minimum degree at least $\delta\Delta/2$.
\end{lemma}
\begin{proof}
Let $v\in Y$, and consider $e_v \cup R_v$. We have $c_v \subseteq e_v \cup R_v$ since
\begin{align}
c_v = e_v + R_v^+ \subseteq e_v \cup R_v^+ \subseteq e_v \cup R_v\, .
\end{align}

Next we decompose
\begin{align}
    e_v \cup R_v = (e_v\backslash R_v) \sqcup R_v\, ,\label{eq:disjoint_union}
\end{align}
so that
\begin{align}
   |c_v| &= |(e_v \cup R_v) \cap c_v| \\
   &= |(e_v\backslash R_v)\cap c_v| + |R_v\cap c_v| \\
   &\le |(e_v\backslash R_v)\cap c_v| + |c_v|/2\, .\label{eq:y_intersect_c_upper_bound}
\end{align}
The first equality follows from the fact that $c_v \subseteq e_v \cup R_v$, the second equality follows from~\eqref{eq:disjoint_union} and the fact that Hamming weights are additive over disjoint unions. The last inequality follows from the low-overlap property. Therefore, we have
\begin{align}
    \deg_{G_{1,y}^\square}(v) &= |y_v|\\
    &=|e_v\backslash R_v|\\
    &\ge |(e_v\backslash R_v) \cap c_v| \\
    &\ge |c_v|/2 \\
    &\ge \delta\Delta/2\, ,
\end{align}
where the last line follows from the minimum distance of $C_1^\perp$, i.e. $\delta \Delta$, and the fact that $c_v\neq 0$ since $v \in Y$.
\end{proof}

Each vertex $v$ of $\mc G_{1,y}^\square$ has a non-trivial nearest codeword $c_v \in C_1^\perp$. To ensure that the individual columns and rows of $c_v$ are themselves close to non-trivial codewords of $C_A$ and $C_B$, we appeal to the robustness of the dual tensor code $C_1^\perp$. Since robustness only applies to codewords of weight at most $\Delta^{3/2+\varepsilon}$, we first define the concept of a \emph{normal} vertex. Roughly speaking, a vertex is considered \emph{normal} precisely when robustness can be applied to its nearest codeword. 

\begin{definition}
    Let us define a \emph{normal} vertex of $Y$ as a vertex with degree at most $\frac{1}{2}\Delta^{3/2+\varepsilon}$ in $\mathcal{G}_{1,y}^\square$. A vertex of $Y$ which is not normal is called \emph{exceptional}. We denote the subsets of normal and exceptional vertices as $Y_n$ and $Y_e$, respectively.
\end{definition}

Since $\mc G^\square_{1,y}$ has large minimum degree, the expansion of $\mc G^\square_1$ now ensures that as long as $\mc G^\square_{1,y}$ has sufficiently few edges, it must contain many normal vertices. Note that Lemma~\ref{lem:y_except_bound} is the only place where the assumption on the weight of $|e|$ (and hence $|y|$) is explicitly used.

\begin{lemma}\label{lem:y_except_bound}
Suppose that $|y| \le \delta n/6\Delta^{3/2-\varepsilon} = \delta\Delta^{1/2+\varepsilon}|V_1|/12$. Then the fraction of exceptional vertices in $Y_e\subseteq Y$ is bounded above as
\begin{align}
    \frac{|Y_e|}{|Y|} \le \frac{576}{\Delta^{1+2\varepsilon}}\, .
\end{align}
\end{lemma}
\begin{proof}
By Lemma~\ref{lem:min_deg}, the minimnum degree of $\mathcal{G}_{1,y}^\square$ is at least $\frac{1}{2}\delta\Delta$. This implies that
\begin{align}
    |Y| \le \frac{2}{\delta\Delta} 2|y| \le \frac{|V_1|}{3\Delta^{1/2-\varepsilon}}\, .\label{eq:Y_bound}
\end{align}

Applying the Expander Mixing Lemma to $E(Y_e,Y)$ in $\mathcal{G}^\square_1$, we get
\begin{align}
    |E(Y_e,Y)| &\le \frac{\Delta^2}{|V_1|}|Y||Y_e| + 4\Delta\sqrt{|Y_e||Y|}\\
    &\le \frac{1}{3}\Delta^{3/2+\varepsilon}|Y_e| + 4\Delta\sqrt{|Y_e||Y|}\, .
\end{align}
By definition of $Y_e$, it holds that $|E(Y_e,Y)| \ge \frac{1}{2}\Delta^{3/2+\varepsilon}|Y_e|$. Combining the inequalities, it follows that
\begin{align}
    \frac{|Y_e|}{|Y|} \le \frac{576}{\Delta^{1+2\varepsilon}}\, .
\end{align}
\end{proof}

Using the robustness of $C_1^\perp$ and the low-overlap property, we can now show that each column and row of $c_v$ for $v \in Y_n$ is indeed close to a codeword of $C_A$ and $C_B$.

\begin{lemma}\label{lem:c_close_col}
Let $v \in Y_n$ be a normal vertex. Then every column (resp. row) of $c_v$ is distance at most $\Delta^{1/2+\varepsilon}/\delta$ from a codeword in $C_A$ (resp. $C_B$). Moreover, $c_v$ contains at least one row or column which is close to a non-zero codeword of $C_A$ or $C_B$.
\end{lemma}
\begin{proof}
By assumption of $v$ being a normal vertex, we know that $|y_v| = |e_v\backslash R_v| \le \frac{1}{2}\Delta^{3/2+\varepsilon}$. From inequality~\eqref{eq:y_intersect_c_upper_bound}, we see that
\begin{align}
    \frac{1}{2}|c_v| \le |(e_v\backslash R_v)\cap c_v| \le |e_v\backslash R_v| \le \frac{1}{2}\Delta^{3/2+\varepsilon}\, .
\end{align}
By the robustness of the dual tensor code $C_1^\perp$, it follows that the support of $c_v$ is concentrated on the union of at most $|c_v|/\delta\Delta \le \Delta^{1/2+\varepsilon}/\delta$ non-zero columns and rows. Using Lemma~\ref{lem:rowcolsum}, we conclude that there exists a decomposition $c_v = \mathbf{c} + \mathbf{r}$, where $\mathbf{c} \in C_A\otimes \mathbb{F}_2^B$ is supported on at most $\Delta^{1/2+\varepsilon}/\delta$ non-zero columns, and where $\mathbf{r} \in \mathbb{F}_2^A\otimes C_B$ is supported on at most $\Delta^{1/2+\varepsilon}/\delta$ non-zero rows. In particular, this implies that each column (resp. row) of $c_v$ is distance at most $\Delta^{1/2+\varepsilon}/\delta$ from a codeword of $C_A$ (resp. $C_B$). Since $c_v$ is non-zero by definition of $Y$, it follows at least one of $\mathbf{c}$ or $\mathbf{r}$ is non-zero, so that at least one column or row is close to a non-zero codeword.   
\end{proof}

Now we are in a position to start the search for our special vertex $v_0 \in V_0$. To that end, we define our analog of ``heavy'' edges in~\cite{LZ22}, which we call ``dense'' edges.

\begin{definition}[Dense Edges]
Let $E_y \subseteq E(\mathcal{G}^\cup)$ be the edges in $\mathcal{G}^\cup$ which are incident to some square in $y$.

We say that an edge $(v,v')\in E_y$, where $v\in V_1$ and $v'\in V_0$, is \emph{dense} if it is incident to at least $\delta\Delta - \Delta^{1/2+\varepsilon}/\delta$ squares of $c_v$.

We then define the vertex set $W\subseteq V_0$ to be the set of all vertices incident to a normal vertex $v\in Y_n$ through a dense edge.
\end{definition}

From the perspective of a vertex $v' \in V_0$, only individual columns and rows of its neighboring nearest codewords $c_v$ are visible. Dense edges are precisely the edges through which $v'$ expects to see non-trivial codewords of $C_A$ or $C_B$. The set $W\subseteq V_1$ defined above can therefore be thought of as the set of ``candidate'' $v_0$'s. We will identify a vertex of $W$ with a linear number of dense edges but a sublinear number of exceptional neighbors in $Y_e$. Such a vertex will allows us to utilize the robustness properties of the local codes.

We first show that each $v' \in W$ must have many neighbors in $Y$ (see Figure~\ref{fig:dense_edge}).

\begin{figure}
    \centering
    \begin{tikzpicture}[scale=1]
\draw[thick] (0,0) circle(2);
\end{tikzpicture}
    \caption{The faces incident to a dense edge $(v, v')$ connecting $v'\in W$ to a normal vertex $v\in Y_n$. Note that $c_v[\cdot,v']$ is close to a $C_A$ codeword.}
    \label{fig:dense_edge}
\end{figure}
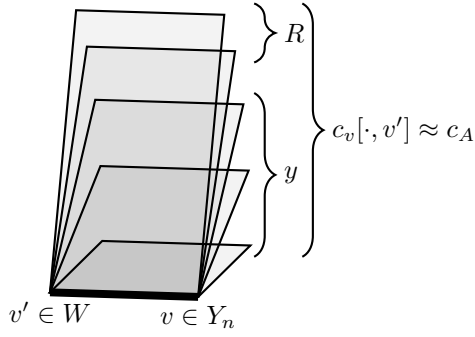

\begin{lemma}\label{lem:deg_W}
The degree in $E_y$ of any $v' \in W$ is at least $\frac{1}{2}\delta\Delta - \Delta^{1/2+\varepsilon}/\delta$. In particular, every $v' \in W$ is adjacent to at least $\frac{1}{2}\delta\Delta - \Delta^{1/2+\varepsilon}/\delta$ vertices in $Y$.
\end{lemma}
\begin{proof}
Let $v'\in W$. By assumption, there exists a dense edge $(v,v')$ connecting $v'$ to a normal vertex $v\in Y_n$. Let us assume without loss of generality that $(v,v')$ is a $B$-edge so that $c_v[\cdot,v']$ defines a column of $c_v$.

Note that the degree of $v'$ in $E_y$ is lower bounded by the weight of the corresponding column in $y_v$, i.e., $\deg_{E_y}(v') \ge |y_v[\cdot,v']|$.

Let $c_A \in C_A$ denote the codeword closest to $c_v[\cdot,v']$. Since $(v,v')$ is dense, it follows from Lemma~\ref{lem:c_close_col} that $c_A$ is non-zero. We can form the matrix which is zero everywhere except on the $v'$-column, where it is equal to $c_A$. Note that this matrix will be a codeword of $C_1^\perp$, and that the low-overlap property applied to this codeword implies that $|R_v[\cdot,v'] \cap c_A| \le |c_A|/2$. 

Then we have
\begin{align}
    |c_A| &= |c_v[\cdot,v']\cap c_A| + |c_A\backslash c_v[\cdot,v']|\\
    &\le |c_v[\cdot,v']\cap c_A| + \Delta^{1/2+\varepsilon}/\delta\\
    &\le |y_v[\cdot,v']\cap c_A| + |R_v[\cdot,v']\cap c_A| + \Delta^{1/2+\varepsilon}/\delta \\
    &\le |y_v[\cdot,v']| + |c_A|/2 + \Delta^{1/2+\varepsilon}/\delta\, ,
\end{align}
where the second line follows from Lemma~\ref{lem:c_close_col}, the third line from the fact that $c_v \subseteq y_v \cup R_v$, and the last line from the low-overlap property. This gives us
\begin{align}
    \delta\Delta/2 \le |c_A|/2 \le |y_v[\cdot,v']| + \Delta^{1/2+\varepsilon}/\delta\, .
\end{align}
Therefore we have
\begin{align}
    \deg_{E_y}(v') \ge |y_v[\cdot,v']| \ge \frac{1}{2}\delta\Delta - \Delta^{1/2+\varepsilon}/\delta\, .
\end{align}
Lemma~\ref{lem:y_subgraph} now ensures that each $v' \in W$ is adjacent to at least $\frac{1}{2}\delta\Delta - \Delta^{1/2+\varepsilon}/\delta$ elements of $Y$. 
\end{proof}

Knowing that each $v'\in W$ has many neighbors in $Y$, the expansion of $\mc G^\cup$ implies that the number of vertices in $W$ must be small compared to $Y$.

\begin{lemma}\label{lem:YW_bound}
For $\Delta$ large enough, the set $W$ satisfies the bound
\begin{align}
    |W|\le \frac{81}{\delta^2\Delta}|Y|\, .
\end{align}
\end{lemma}
\begin{proof}
Using Lemma~\ref{lem:deg_W}, we know that each vertex in $W$ is adjacent to at least $\frac{1}{2}\delta\Delta - \Delta^{1/2+\varepsilon}/\delta$ vertices in $Y$. Therefore we can bound the edges in $\mathcal{G}^\cup$ between $Y$ and $W$ by
\begin{align}
    |E_{\mathcal{G}^\cup}(Y,W)| \ge \left(\frac{1}{2}\delta\Delta - \frac{\Delta^{1/2+\varepsilon}}{\delta}\right)|W| = \frac{1}{2}\delta\Delta\left(1 - \frac{2}{\delta^2\Delta^{1/2-\varepsilon}}\right)|W|\, .
\end{align}
Applying the Expander Mixing Lemma, we have
\begin{align}
    |E_{\mathcal{G}^\cup}(Y,W)| \le \frac{\Delta}{|V_1|}|Y||W|+4\Delta^{1/2}\sqrt{|Y||W|}\, .
\end{align}
From equation~\eqref{eq:Y_bound}, we have
\begin{align}
    |Y| \le \frac{|V_1|}{3\Delta^{1/2-\varepsilon}}\, .
\end{align}
Combining these inequalities, we end up with
\begin{align}
    \frac{1}{2}\delta\Delta\left(1 - \frac{2}{\delta^2\Delta^{1/2-\varepsilon}}\right)|W| &\le \frac{\Delta}{|V_1|}|Y||W|+4\Delta^{1/2}\sqrt{|Y||W|}\\
    &\le \frac{1}{3}\Delta^{1/2+\varepsilon}|W|+4\Delta^{1/2}\sqrt{|Y||W|}\, ,
\end{align}
or equivalently,
\begin{align}
    \frac{1}{8}\delta\Delta^{1/2}\left(1 - \frac{2}{\delta^2\Delta^{1/2-\varepsilon}} - \frac{2}{3\delta\Delta^{1/2-\varepsilon}}\right) \le \sqrt{\frac{|Y|}{|W|}}\, .
\end{align}
Taking $\Delta$ sufficiently large so that
\begin{align}
   1 - \frac{2}{\delta^2\Delta^{1/2-\varepsilon}} - \frac{2}{3\delta\Delta^{1/2-\varepsilon}} \ge \frac{8}{9}\, ,
\end{align}
we end up with the desired bound.
\end{proof}

We expect each $v\in Y_n$ to be incident to at least one dense edge by virtue of having a column or row close to a non-trivial codeword. This means that the total number of dense edges is at least on the order of $|Y_n|$. Lemma~\ref{lem:YW_bound} in turn suggests that the number of dense edges is large relative to $|W|$. This implies that the average vertex in $W$ should be incident to a large number of dense edges. This is formalized by Lemma~\ref{lem:dense_deg} and Corollary~\ref{cor:dense_frac} below.

\begin{lemma}\label{lem:dense_deg}
Let $\mathcal{D}$ denote the set of dense edges incident to $W$. Then the average degree of $W$ in $\mathcal{D}$ is bounded by
\begin{align}
    \frac{\mathcal{|D|}}{|W|} \ge 2\alpha \Delta
\end{align}
for some constant $\alpha > 0$.\footnote{Note that we may choose $\alpha$ to be anything smaller than $\delta^2/192$ by taking $\Delta$ sufficiently large.}
\end{lemma}
\begin{proof}
First, note that every $v\in Y_n$ is incident to at least one dense edge, which is then by definition in $\mathcal{D}$. To see this, consider $c_v$, which is non-zero by definition of $Y$. It follows from Lemma~\ref{lem:c_close_col} that $c_v$ contains at least one column or row which is close to a non-zero codeword of $C_A$ or $C_B$, which in turn implies that column or row must have weight at least $\delta\Delta - \Delta^{1/2+\varepsilon}/\delta$. By definition, such a column or row is defined by some edge $(v,v') \in \mathcal{G}^\cup$, which is then a dense edge incident to $v$.

Since each dense edge has at most one endpoint in $Y_n$, it follows the above discussion that $|\mathcal{D}| \ge |Y_n| = |Y|-|Y_e|$. From Lemmas~\ref{lem:y_except_bound} and~\ref{lem:YW_bound}, it follows that
\begin{align}
    |Y|-|Y_e| \ge \left(1 - \frac{576}{\Delta^{1+2\varepsilon}}\right)|Y| \ge \frac{\Delta\delta^2}{81}\left(1 - \frac{576}{\Delta^{1+2\varepsilon}}\right)|W|\, .
\end{align}
Therefore we get
\begin{align}
\frac{|\mathcal{D}|}{|W|} \ge \frac{\delta^2}{81}\left(1 - \frac{576}{\Delta^{1+2\varepsilon}}\right)\Delta \equiv 2\alpha\Delta\, .
\end{align}
\end{proof}

\begin{corollary}\label{cor:dense_frac}
At least an $\alpha/2$ fraction of the vertices in $W$ are incident to at least $\alpha\Delta$ dense edges.
\end{corollary}
\begin{proof}
Let $\eta$ be the fraction of vertices in $W$ with dense degree greater than $\alpha\Delta$. The maximum degree of any vertex in $\mathcal{G}^\cup$ is $2\Delta$, so it follows that
\begin{align}
    2\alpha\Delta \le \frac{|\mathcal{D}|}{|W|} \le 2\Delta\eta + (1-\eta)\alpha\Delta\, .
\end{align}
Therefore we have $\eta \ge \alpha/(2-\alpha) \ge \alpha/2$.
\end{proof}

We have now shown that there exists a subset of vertices in $W$ incident to many dense edges. We must now show that within this subset, there exists vertices which are \emph{not} adjacent to many exceptional vertices in $Y_e$. We expect this to be the case since the number of exceptional vertices is small relative to the number of normal vertices. To proceed, we bound the number of edges shared between $W$ and $Y_e$ in Lemma~\ref{lem:WYe_edge} below.

\begin{lemma}\label{lem:WYe_edge}
The total number of edges in $\mathcal{G}^\cup$ between $W$ and $Y_e$ is bounded above by
\begin{align}
    |E_{\mathcal{G}^\cup}(W,Y_e)| \le 193\Delta^{1/2-\varepsilon}|W|\, .
\end{align}
\end{lemma}
\begin{proof}
Using the Expander Mixing Lemma, we get
\begin{align}
|E_{\mathcal{G}^\cup}(W,Y_e)| \le \frac{\Delta}{|V_1|}|Y_e||W| + 4\sqrt{\Delta}\sqrt{|Y_e||W|}\, .
\end{align}
Using Lemma~\ref{lem:y_except_bound} and inequality~\eqref{eq:Y_bound}, this becomes
\begin{align}
    |E_{\mathcal{G}^\cup}(W,Y_e)| &\le \frac{576}{|V_1|\Delta^{2\varepsilon}}|Y||W| + 96\Delta^{-\varepsilon}\sqrt{|Y||W|}\\
    &\le \frac{192}{\Delta^{1/2+\varepsilon}}|W| + 96\Delta^{-\varepsilon}\sqrt{|Y||W|}\, .
\end{align}
As noted in the proof of Lemma~\ref{lem:dense_deg}, each vertex of $Y_n$ is incident to at least one vertex in $W$. Since each vertex of $W$ has degree $2\Delta$, it follows that $|Y_n| \le 2\Delta|W|$. Choosing $\Delta$ sufficiently large that
\begin{align}
    \frac{576}{\Delta^{1+2\varepsilon}} \le \frac{1}{2}\, ,\label{eq:y_exp_half}
\end{align}
it follows from Lemma~\ref{lem:y_except_bound} that $|Y_n| = |Y|-|Y_e| \ge |Y|/2$, so that $|Y| \le 4\Delta|W|$. Combining these bounds, we obtain
\begin{align}
      |E_{\mathcal{G}^\cup}(W,Y_e)| &\le  \frac{192}{\Delta^{1/2+\varepsilon}}|W| + 96\Delta^{-\varepsilon}\sqrt{|Y||W|}\\
      &\le \frac{192}{\Delta^{1/2+\varepsilon}}|W| + 192\Delta^{1/2-\varepsilon}|W|\\
      &= 192\left(1+\frac{1}{\Delta}\right)\Delta^{1/2-\varepsilon}|W|\\
      &\le 193\Delta^{1/2-\varepsilon}|W|\, .
\end{align}
\end{proof}

Putting everything together, we can finally show the existence of the special vertex $v_0$, as formalized by Corollary~\ref{cor:special_vertex}.

\begin{corollary}\label{cor:special_vertex}
At least an $\alpha/4$ fraction of the vertices of $W$:
\begin{enumerate}
    \item are incident to at least $\alpha\Delta$ dense edges, and
    \item are adjacent to at most $(772/\alpha) \Delta^{1/2-\varepsilon} \equiv \beta\Delta^{1/2-\varepsilon}$ vertices of $Y_e$.
\end{enumerate}
In particular, at least one such vertex exists since $\alpha > 0$.
\end{corollary}
\begin{proof}
Let $W_1$ be the subset of vertices in $W$ satisfying condition 1, and let $\overline{W_2}$ be the subset of vertices in $W$ \emph{not} satisfying condition 2. Since each vertex of $\overline{W_2}$ is adjacent to more than $(772/\alpha)\Delta^{1/2-\varepsilon}$ vertices of $Y_e$, we get 
\begin{align}
|\overline{W_2}|\cdot (772/\alpha) \Delta^{1/2-\varepsilon} \le |E_{\mathcal{G}^\cup}(W,Y_e)| \le 193\Delta^{1/2-\varepsilon}|W|\, ,
\end{align}
which implies that $|\overline{W_2}| \le (\alpha/4)|W|$. Therefore the set of vertices satisfying both condition 1 and 2 is bounded below by
\begin{align}
    |W_1\backslash \overline{W_2}| \ge |W_1| - |\overline{W_2}| \ge \alpha|W|/2 - \alpha|W|/4 = \alpha|W|/4\, .
\end{align}
\end{proof}

\subsubsection{The local view at $v_0$}

Let $v_0 \in W$ be a vertex satisfying the conditions of Corollary~\ref{cor:special_vertex}. In this subsection, we analyze the structure of $y$ and $R$ \emph{from the perspective of $v_0 \in V_0$}. Let $y_0$, $e_0$, and $R_0$ denote the local views of $y$, $e$, and $R$ at the vertex $v_0$. 

We will write $[v,v'] \in Q(v_0)$ to denote the face anchored at $v_0$ with neighboring $V_1$ vertices $v$ and $v'$, with the implicit convention that unprimed vertices $v$ denote row vertices, and primed vertices $v'$ denote column vertices. We will also write $N(v_0)\subseteq V_1$ to denote the set of all neighbors of $v_0$ in $\mc G^\cup$, and $N_r(v_0)$ and $N_c(v_0)$ to denote the set of row and column vertex neighbors, respectively.

We first show a key result regarding the structure of $y_0$ and $R_0$. As a consequence of metastability, the edges of $R_0$ must complement the edges of $y_0$ to complete codewords on either columns or rows shared with neighboring local views (see equation~\ref{eq:y_R_decomp}). This allows us to split $R_0$ into disjoint parts depending on whether columns or rows are corrected.

\begin{lemma}\label{lem:R_decomp}
We can write $R_0 = R_\mathrm{col} \sqcup R_\mathrm{row}$, where we have
\begin{align}
    y_0[v,\cdot]\sqcup R_{\mathrm{row}}[v,\cdot] = c_v[v_0,\cdot]\, ,\qquad\text{and}\qquad y_0[\cdot,v']\sqcup R_{\mathrm{col}}[\cdot,v'] = c_{v'}[\cdot,v_0]\, ,\label{eq:y_R_decomp}
\end{align}
for all $v \in N_r(v_0)$ and $v' \in N_c(v_0)$.
\end{lemma}
\begin{proof}
Let $q = [v,v'] \in R_0$. Since $R$ is metastable, it follows that $q$ belongs to exactly one of $R_v^+$ or $R_{v'}^+$. Suppose without loss of generality that $q\in R_v^+$. Since $e_v + R_v^+ = c_v$, it follows that $q \in c_v$ if and only if $q \notin e$. Likewise, since $q \notin R_{v'}^+$, it follows that $q \in c_{v'}$ if and only if $q \in e$. It follows that $q$ must be an element of exactly one of $c_v$ or $c_{v'}$.

Let $R_{\mathrm{row}}\subseteq R_0$ denote the collection of all $q\in R_0$ which belong to $c_v$ for some row vertex $v$. Likewise, let $R_{\mathrm{col}}\subseteq R_0$ denote the collection of all $q\in R_0$ which belong to $c_{v'}$ for some column vertex $v'$. Then by the preceding discussion we have 
\begin{align}
R_0 = R_{\mathrm{row}} \sqcup R_{\mathrm{col}}\, .    
\end{align}

Next, we show equation~\eqref{eq:y_R_decomp}. We focus on the row case, with the column case being analogous. Note that we have 
\begin{align}
y_0[v,\cdot] = e_v[v_0,\cdot]\backslash R_v[v_0,\cdot] \subseteq e_v[v_0,\cdot]\backslash R_v^+[v_0,\cdot] \subseteq e_v[v_0,\cdot] + R_v^+[v_0,\cdot] = c_v[v_0,\cdot]\,.
\end{align}
Also, we have $R_{\mathrm{row}}[v,\cdot] \subseteq c_v[v_0,\cdot]$ by definition. This implies that
\begin{align}
    y_0[v,\cdot] \sqcup R_{\mathrm{row}}[v,\cdot] \subseteq c_v[v_0,\cdot]\, .
\end{align}
Conversely, we have 
\begin{align}
c_v[v_0,\cdot] = e_v[v_0,\cdot] + R^+_v[v_0,\cdot] \subseteq e_v[v_0,\cdot] \cup R_v[v_0,\cdot] = y_v[v_0,\cdot]\sqcup R_v[v_0,\cdot] = y_0[v,\cdot] \sqcup R_0[v,\cdot]\, .
\end{align} 
Since all elements of $R_0$ belonging to $c_v$ are by definition in $R_{\mathrm{row}}$, it follows that we have
\begin{align}
    c_v[v_0,\cdot] \subseteq y_0[v,\cdot] \sqcup R_{\mathrm{row}}[v,\cdot]\, .
\end{align}
It therefore follows that
\begin{align}
    y_0[v,\cdot]\sqcup R_{\mathrm{row}}[v,\cdot] = c_v[v_0,\cdot]\, ,\qquad\text{and}\qquad y_0[\cdot,v']\sqcup R_{\mathrm{col}}[\cdot,v'] = c_{v'}[\cdot,v_0]\, ,
\end{align}
which hold for all $v \in N_r(v_0)$ and $v'\in N_c(v_0)$.
\end{proof}

\begin{corollary}\label{cor:R_notin_Y}
Let $[v,v']\in Q(v_0)$. If $v\notin Y$ then $R_{\mathrm{row}}[v,\cdot] = 0$. Likewise, if $v'\notin Y$ then $R_{\mathrm{col}}[\cdot,v'] = 0$.
\end{corollary}
\begin{proof}
We work with the row vertex $v$, with the column case being identical. Suppose that $v\notin Y$. Then by definition, the closest codeword to $e_v$ at $v$ is the trivial codeword $c_v = 0$. Evaluating equation~\eqref{eq:y_R_decomp} at the row defined by edge $(v_0,v)$, we have
\begin{align}
    y_0[v,\cdot]\sqcup R_{\mathrm{row}}[v,\cdot] = c_v[v_0,\cdot] = 0\, , 
\end{align}
which implies that $R_{\mathrm{row}}[v,\cdot] = 0$.
\end{proof}

Let us now provide some intuition for the remainder of the proof. The decomposition shown in Lemma~\ref{lem:R_decomp} allow us to consider two separate scenarios: 
\begin{enumerate}
\item First, imagine that $R_0$ has high weight relative to $y_0$. Then Lemma~\ref{lem:R_decomp} suggests that the columns and rows of $R_0$ are close to codewords of $C_A$ and $C_B$. An argument similar to the one used in the proof of case 2 would seem to suggest that there exists some subset of $R_0$ which would decrease the global potential when flipped. 

\item Alternatively, consider the scenario where $R_0$ has low weight relative to $y_0$. In this case, $y_0$ is close to $e_0$, and Lemma~\ref{lem:R_decomp} now implies that the columns and rows of $e_0$ are close to codewords of $C_A$ and $C_B$. The robustness of the dual tensor code $C_1^\perp$ suggests that we can find a codeword $c_0 \in C_A\otimes C_B$, i.e., a $Z$-stabilizer, which has high overlap with $e_0$. But this is in contradiction with the fact that $e$ was assumed to be a reduced error.
\end{enumerate}

Given the discussion above, we will finish the proof as follows: Suppose that no subset of $Q(v_0)$ decreases the global potential when flipped. We will show that this necessarily implies that $R_0$ has sufficiently low weight (as formalized by Lemma~\ref{lem:R_low_weight}) that the argument outlined in scenario 2 can be carried out. Specifically, we will show that there exists some $c_0 \in C_A\otimes C_B$ such that $|e+c_0| < |e|$, contradicting the fact that $e$ is reduced.

To proceed, we will need to analyze the value of the potential on a new configuration of errors, one obtained from $e$ by flipping all the qubits of $e\cap R_0$. The utility of this new error configuration $\tilde{e}$ comes from the fact that the rows of $R_{\mathrm{row}}$ and columns of $R_{\mathrm{col}}$ are exactly equal to the local minimum weight corrections for $\tilde{e}$ (see equation~\ref{eq:rowcol_correction}), giving us better control over the potential.

Let $\tilde{e} = e + e\cap R_0 = e\backslash R_0$. We first show that some key quantities remain unchanged in this new error configuration. Since $\tilde{e}$ is obtained from $e$ by flipping a subset of $R$ without decreasing the global potential, the $R$-flipping Lemma~\ref{lem:R_flip} implies that the new total correction $\tilde{R} \equiv R(\tilde{e})$ will be equal to the old one, i.e., $\tilde{R} = R(\tilde{e}) = R(e)$. This implies that the vector of residual errors $y$ likewise stays invariant, i.e., $\tilde{y} = y(\tilde{e}) = \tilde{e}\backslash \tilde{R} = e\backslash R(e) = y(e)$. The situation after flipping $R_0 \cap e$ is illustrated in Figure~\ref{fig:lemma40_flip} and summarized by Lemma~\ref{lem:local_update}.

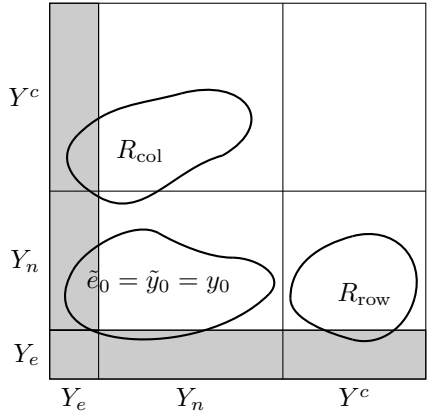
\begin{figure}
    \centering
    \begin{tikzpicture}[scale=1]
\draw[thick] (0,0) circle(2);
\end{tikzpicture}
    \caption{The $v_0$ local view after flipping $R_0\cap e$. The various regions indicate the possible supports of the labeled quantities.}
    \label{fig:lemma40_flip}
\end{figure}

\begin{lemma}\label{lem:local_update}
Suppose that no subset of $Q(v_0)$ decreases the global potential when flipped. Let $\tilde{e} = e\backslash R_0$ denote the configuration of errors obtained after flipping all the elements of $R_0\cap e$. In this new error configuration, we may take the local minimum weight corrections to be as given by the $R$-flipping Lemma~\ref{lem:R_flip}. Specifically, we have $\tilde{R} = R$ and $\tilde{y} \equiv \tilde{e}\backslash \tilde{R} = e\backslash R = y$. Moreover, we have $\tilde{e}_0 = y_0$, and
\begin{align}
    R_{\mathrm{row}}[v,\cdot] = \tilde{R}_v^+[v_0,\cdot]\, ,\quad\text{and}\quad R_{\mathrm{col}}[\cdot, v'] = \tilde{R}^+_{v'}[\cdot,v_0]\, ,\label{eq:rowcol_correction}
\end{align}
for all $[v,v'] \in Q(v_0)$.
\end{lemma}
\begin{proof}
The fact that we may take $\tilde{R} = R$ follows directly from the $R$-flipping Lemma~\ref{lem:R_flip}, which ensures that the original and updated local minimum weight correction sets differ only by the orientations of edges. It follows that we also have
\begin{align}
    y = e\backslash R = (e\backslash R_0)\backslash R = \tilde{e}\backslash \tilde{R} = \tilde{y}\, .
\end{align}
Note that since $\tilde{e}\cap R_0 = \emptyset$, it also follows that $\tilde{y}_0 = \tilde{e}_0$.

Now, let $v$ be a neighbor of $v_0$, and suppose without loss of generality that it is a row vertex. By the $R$-flipping Lemma~\ref{lem:R_flip}, the nearest codeword $c_v$ remains unchanged after flipping $R_0\cap e$. In particular, we must have
\begin{align}
     y_0[v,\cdot] \sqcup R_{\mathrm{row}}[v,\cdot] = c_v[v_0,\cdot] = \tilde{e}_0[v,\cdot] + \tilde{R}_v^+[v_0,\cdot] = y_0[v,\cdot] \sqcup \tilde{R}_v^+[v_0,\cdot]\, , 
\end{align}
where the first equality follows from Lemma~\ref{lem:R_decomp}, the second from the invariance of the codeword $c_v$, and the last from the facts that $\tilde{e}_0 = \tilde{y}_0 = y_0$ and $\tilde{e}_0[v,\cdot]\cap \tilde{R}^+_v[v_0,\cdot] \subseteq \tilde{e}_0[v,\cdot]\cap \tilde{R}_0[v,\cdot] = \emptyset$. It follows that we must have $R_{\mathrm{row}}[v_,\cdot] = \tilde{R}_v^+[v_0,\cdot]$.
\end{proof}

Since the rows (resp. columns) of $R_{\mathrm{row}}$ (resp. $R_{\mathrm{col}}$) are equal to the local minimum weight corrections (for $\tilde{e}$) on neighboring vertices, we expect that $R_0$ cannot be too large. Otherwise, $R_0$ would have enough overlap with the neighboring local minimum weight corrections that subsets of it can start lowering the potential. Therefore the fact that no subset of $R_0$ can lower the potential implicitly places a bound on its size. This is formalized by Lemma~\ref{lem:R_low_weight} below.

\begin{lemma}\label{lem:R_low_weight}
Suppose that no subset of $Q(v_0)$ decreases the global potential $U$ when flipped. Then we have 
\begin{align}
    |R_0| \le \frac{3\Delta^{3/2+\varepsilon}}{\delta}
\end{align}
for sufficiently large $\Delta$.
\end{lemma}

\begin{proof}
Consider the error configuration $\tilde{e} = e\backslash R_0$. By assumption we have $U(\tilde{e}) = U(e)$. Using Lemma~\ref{lem:local_update}, we have $\tilde{e}_0 = \tilde{y}_0 = y_0$ and $\tilde{R}_0 = R_0$.

Let $v$ be, without loss of generality, a row vertex. Since we have $R_{\mathrm{row}}[v,\cdot] = \tilde{R}_v^+[v_0,\cdot]$,
it follows that flipping $R_{\mathrm{row}}[v,\cdot]$ decreases the local potential $U_v(\tilde{e})$ by $|R_{\mathrm{row}}[v,\cdot]|$, i.e.,
\begin{align}
    U_v(\tilde{e} + R_{\mathrm{row}}[v,\cdot]) = U_v(\tilde{e}) - |R_{\mathrm{row}}[v,\cdot]|\, .\label{eq:potential_reduce}
\end{align}
Now, suppose that $v \in Y_n$. Let $c_B$ be the closest codeword of $C_B$ to $c_v[v_0,\cdot]$. Then
\begin{align}
    U_v(\tilde{e}+y_0[v,\cdot]) &= U_v(\tilde{e}+R_{\mathrm{row}}[v,\cdot] + c_v[v_0,\cdot])\\ 
    &\le U_v(\tilde{e}+R_{\mathrm{row}}[v,\cdot] + c_B) + \frac{\Delta^{1/2+\varepsilon}}{\delta}\\
    &=U_v(\tilde{e}+R_{\mathrm{row}}[v,\cdot]) + \frac{\Delta^{1/2+\varepsilon}}{\delta}\\
    &=U_v(\tilde{e})-|R_{\mathrm{row}}[v,\cdot]| + \frac{\Delta^{1/2+\varepsilon}}{\delta}\, ,\label{eq:potential_reduce2}
\end{align}
where the first equality follows from the fact that
\begin{align}
y_0[v,\cdot] \sqcup R_{\mathrm{row}}[v,\cdot] =
y_0[v,\cdot] + R_{\mathrm{row}}[v,\cdot] = c_v[v_0,\cdot].
\end{align}
The second line follows from Lemma~\ref{lem:c_close_col}, and the third line follows from the fact that $U_v(e+c) = U_v(e)$ for any $c \in C_1^\perp$. The last line is just equation~\eqref{eq:potential_reduce}. Note that an analogous version of inequality~\eqref{eq:potential_reduce2} also holds for column vertices.

Consider now the global potential $U(\tilde{e} + y_0)$. Note that it follows from Lemma~\ref{lem:y_subgraph} that $y_0$ will have empty intersection with the local view of any $v$ not in $Y$, so that only the local potentials associated with vertices of $Y$ can be affected by flipping $y_0$. We will bound the potential by explicitly separating out the contributions of the exceptional vertices in $Y_e$ over which we have little control. Let us write $\beta \equiv 772/\alpha$ for the constant appearing in Corollary~\ref{cor:special_vertex}. Then we can bound the change in the potential by
\begin{align}
    0 &\le U(\tilde{e}+y_0) - U(\tilde{e})\\ 
    &= \sum_{v \in N(v_0)\cap Y}(U_v(\tilde{e}+y_0) - U_v(\tilde{e}))\\
    &\le \sum_{v \in N(v_0)\cap Y_n}(U_v(\tilde{e}+y_0) - U_v(\tilde{e})) + \beta\Delta^{3/2-\varepsilon}\, ,
\end{align}
where the first inequality follows from the assumption that no subset of $Q(v_0)$ decreases the global potential when flipped, the second line from the fact that only the local views associated with vertices of $N(v_0)\cap Y$ are affected by flipping $y_0$, and the last line removes the contributions resulting from the vertices in $Y_e$. The $\beta\Delta^{3/2-\varepsilon}$ term in the last line comes from the fact that there are at most $\beta\Delta^{1/2-\varepsilon}$ vertices of $N(v_0)\cap Y_e$ as a result of Corollary~\ref{cor:special_vertex}, each of which can increase the weight of the potential by at most $\Delta$.

Splitting the sum above into row and column parts and applying inequality~\eqref{eq:potential_reduce2}, we get
\begin{align}
    &\sum_{v \in N(v_0)\cap Y_n}(U_v(\tilde{e}+y_0) - U_v(\tilde{e})) \\
    = &\sum_{v \in N_r(v_0)\cap Y_n}(U_v(\tilde{e}+y_0[v,\cdot]) - U_v(\tilde{e})) + \sum_{v' \in N_c(v_0)\cap Y_n}(U_{v'}(\tilde{e}+y_0[\cdot,v']) - U_{v'}(\tilde{e}))\\
    \le &\sum_{v \in N_r(v_0)\cap Y_n}\left(-|R_\mathrm{row}[v,\cdot]| + \frac{\Delta^{1/2+\varepsilon}}{\delta}\right) + \sum_{v' \in N_c(v_0)\cap Y_n}\left(-|R_{\mathrm{col}}[\cdot,v']| + \frac{\Delta^{1/2+\varepsilon}}{\delta}\right)\\
    \le &-\sum_{v \in N_r(v_0)\cap Y_n}|R_{\mathrm{row}}[v,\cdot]| - \sum_{v' \in N_c(v_0)\cap Y_n}|R_{\mathrm{col}}[\cdot,v']|  + \frac{2\Delta^{3/2+\varepsilon}}{\delta}\, .
\end{align}

By Corollary~\ref{cor:R_notin_Y}, it follows that the rows of $R_{\mathrm{row}}$ (and columns of $R_{\mathrm{col}}$, respectively) are zero if the indexing vertex is not in $Y$. It follows that we have
\begin{align}
    \sum_{v \in N_r(v_0)\cap Y_n}|R_{\mathrm{row}}[v,\cdot]| = \sum_{v \in N_r(v_0)\backslash Y_e}|R_{\mathrm{row}}[v,\cdot]| \ge |R_{\mathrm{row}}| - \beta\Delta^{3/2-\varepsilon}\, ,
\end{align}
and likewise
\begin{align}
    \sum_{v' \in N_c(v_0)\cap Y_n}|R_{\mathrm{col}}[\cdot,v']| = \sum_{v' \in N_c(v_0)\backslash Y_e}|R_{\mathrm{col}}[\cdot,v']| \ge |R_{\mathrm{col}}| - \beta\Delta^{3/2-\varepsilon}\, ,
\end{align}
where the $\beta\Delta^{3/2-\varepsilon}$ correction term again comes from the vertices in $Y_e$ over which we have no control. Altogether, we have
\begin{align}
    0 \le -|R_{\mathrm{row}}| - |R_{\mathrm{col}}| + \frac{2\Delta^{3/2+\varepsilon}}{\delta} + 3\beta\Delta^{3/2-\varepsilon}\, .
\end{align}
Taking $\Delta$ sufficiently large so that $\Delta^{2\varepsilon} \ge 3\delta\beta$, we finally get
\begin{align}
    |R_0| \le \frac{3\Delta^{3/2+\varepsilon}}{\delta}\, .
\end{align}
\end{proof}

Lemma~\ref{lem:R_low_weight} shows that $R_0$ is small. This now allows us to follow the remaining steps outlined in scenario 2 above to complete the proof of Theorem~\ref{thm:correction_exists}.

\begin{corollary}\label{cor:y_distance}
Suppose that no subset of $Q(v_0)$ decreases the global potential $U$ when flipped. Then we have
\begin{align}
    d(y_0,C_A\otimes \mathbb{F}_2^\Delta) + d(y_0, \mathbb{F}_2^\Delta\otimes C_B) \le \frac{10\Delta^{3/2+\varepsilon}}{\delta}
\end{align}
for sufficiently large $\Delta$.
\end{corollary}
\begin{proof}
Consider the distance of $y_0$ to the row codespace $\mathbb{F}_2^A\otimes C_B$ (with the column case being identical). From equation~\eqref{eq:y_R_decomp}, we have
\begin{align}
    y_0[v,\cdot] + R_{\mathrm{row}}[v,\cdot] = y_0[v,\cdot] \sqcup R_{\mathrm{row}}[v,\cdot] = c_v[v_0,\cdot]\, .
\end{align}
If $v \notin Y$ then Corollary~\ref{cor:R_notin_Y} implies that each of the terms above is zero. If $v\in Y_n$, then Lemma~\ref{lem:c_close_col} implies that
\begin{align}
    d(y_0[v,\cdot] + R_{\mathrm{row}}[v,\cdot], C_B) = d(c_v[v_0,\cdot],C_B) \le \frac{\Delta^{1/2+\varepsilon}}{\delta}\, .
\end{align}
Summing over all rows, and accounting for the exceptional vertices $v\in Y_e$, we get
\begin{align}
d(y_0 + R_{\mathrm{row}}, \mathbb{F}_2^A\otimes C_B) \le \frac{\Delta^{3/2+\varepsilon}}{\delta} + \beta\Delta^{3/2-\varepsilon}\, ,    
\end{align}
where the $\Delta^{3/2+\varepsilon}$ term comes from the non-exceptional vertices and the $\Delta^{3/2-\varepsilon}$ term from the exceptional vertices. Since 
\begin{align}
|R_{\mathrm{row}}| \le |R_0| \le \frac{3\Delta^{3/2+\varepsilon}}{\delta}
\end{align}
by Lemma~\ref{lem:R_low_weight}, it follows that we have
\begin{align}
d(y_0, \mathbb{F}_2^A\otimes C_B) \le \frac{4\Delta^{3/2+\varepsilon}}{\delta} + \beta\Delta^{3/2-\varepsilon} \le \frac{5\Delta^{3/2+\varepsilon}}{\delta}\, ,    
\end{align}
where the last inequality follows from the fact that we took $\Delta$ large enough so that $\Delta^{2\varepsilon}\ge 3\beta\delta$ in Lemma~\ref{lem:R_low_weight}.
\end{proof}

\begin{corollary}
Suppose no subset of $Q(v_0)$ decreases the global potential $U$ when flipped. Then the local view $y_0$ has weight 
\begin{align}
    |y_0| \ge \frac{1}{4}\alpha\delta\Delta^2
\end{align}
for sufficiently large $\Delta$.
\end{corollary}
\begin{proof}
From Corollary~\ref{cor:special_vertex} it follows that $v_0$ is adjacent to either $\ge (\alpha\Delta-\beta\Delta^{1/2-\varepsilon})/2$ normal row vertices $v\in N_r(v_0)\cap Y_n$ or $\ge (\alpha\Delta-\beta\Delta^{1/2-\varepsilon})/2$ normal column vertices $v'\in N_c(v_0)\cap Y_n$ through dense edges. Suppose without loss of generality that it is the former. Then by definition of dense edges, it follows that $|c_v[v_0,\cdot]|\ge \delta\Delta - \Delta^{1/2+\varepsilon}/\delta$ for each of these vertices. 

Summing the first equation in~\eqref{eq:y_R_decomp} over all row vertices $v$, we get 
\begin{align}
    |y_0| + |R_{\mathrm{row}}| &= \sum_{v \in N_r(v_0)}|y_0[v,\cdot]\sqcup R_{\mathrm{row}}[v,\cdot]|\\
    &= \sum_{v \in N_r(v_0)}|c_v[v_0,\cdot]|\\
    &\ge (\alpha\Delta - \beta\Delta^{1/2-\varepsilon})(\delta\Delta - \Delta^{1/2+\varepsilon}/\delta)/2\, ,
\end{align}
where the last inequality follows from the preceding discussion. Choosing $\Delta$ sufficiently large so that
\begin{align}
    (\alpha\Delta - \beta\Delta^{1/2-\varepsilon})(\delta\Delta - \Delta^{1/2+\varepsilon}{\delta}) \ge \frac{2}{3}\alpha\delta\Delta^2
\end{align}
and applying Lemma~\ref{lem:R_low_weight}, we get
\begin{align}
    |y_0| \ge \frac{1}{3}\alpha\delta\Delta^2 - \frac{3\Delta^{3/2+\varepsilon}}{\delta}\, .
\end{align}
This implies that
\begin{align}
    |y_0| \ge \frac{1}{4}\alpha\delta\Delta^2\, ,
\end{align}
again for sufficiently large $\Delta$.
\end{proof}

Finally, we are now in a position to complete the proof of Theorem~\ref{thm:correction_exists}. 

\begin{proof}[Theorem~\ref{thm:correction_exists}]
Since the code $C_1^\perp$ is chosen to be $\Delta^{3/2+\varepsilon'}$ robust for $\varepsilon' > \varepsilon$, it follows from Corollary~\ref{cor:y_distance} and Proposition~\ref{prop:robustness} that there exists some $c_0 \in C_A\otimes C_B$ such that $|y_0-c_0| \le 15\Delta^{3/2+\varepsilon}/\delta$, which holds so long as $\Delta$ is chosen large enough so that $\delta\Delta^{\varepsilon'} \ge 10\Delta^{\varepsilon}$. Applying Lemma~\ref{lem:R_low_weight}, this implies that
\begin{align}
    |e_0+c_0| = |y_0 + e_0\cap R_0 + c_0| \le |y_0 + c_0| + |R_0\cap e_0| \le \frac{18\Delta^{3/2+\varepsilon}}{\delta}\, .
\end{align}
Since we have $|e_0| \ge |y_0|\ge (\alpha\delta/4)\Delta^2$, it follows that we have $|e_0 + c_0| < |e_0|$ whenever
\begin{align}
   \frac{72}{\alpha\delta^2} < \Delta^{1/2-\varepsilon}\, .
\end{align}
This contradicts the fact that $e$ was chosen to be a reduced error.
\end{proof}

\section{Discussion and Conclusion}\label{sec:conclusion}

In this paper, we have shown the existence of a provably correct decoder for the recent quantum Tanner codes construction of asymptotically good qLDPC codes. Our decoder has runtime linear in the code blocklength, and provably corrects all errors with weight up to a constant fraction of the distance (and hence the blocklength). A key idea behind the decoder is the introduction of a global potential function which measures the stability of the error against locally defined corrections. Our decoder proceeds operationally in a manner similar to the small-set-flip decoder for quantum expander codes~\cite{LTZ15}, checking candidate subsets defined within the local views of the code to see if the global potential function can be reduced at each step. We prove that such a reduction is always possible for sufficiently low weight errors, which we use to show that the decoder successfully corrects all errors of weight $|e| \lesssim \delta n/\Delta^{7/2+\varepsilon}$. The existence of our decoder implies a notion of soundness for the quantum Tanner codes construction (see Corollary~\ref{cor:soundness}). It also implies an accuracy threshold against stochastic noise (see Corollary~\ref{cor:threshold}).

An important part of our proof for the correctness of the decoder involves showing the existence of dual tensor codes of larger robustness ($\Delta^{3/2+\varepsilon}$) than was established in~\cite{LZ22}. This result also gives a constant factor improvement in the distance of the code. In addition, it leads to a simplification in the construction of quantum Tanner codes in that the dual tensor codes are no longer required to be resistant to puncturing.

A number of open problems remain at this point. One major problem is the time-complexity of the decoder. While the runtime of the decoder is linear in the blocklength, there are constant prefactors on the order of $2^{\Delta^2}$ arising from the need to check all subsets of the $\Delta^2$-sized local views. This renders the decoder impractical in reality. Part of the problem stems from the inherently large check weights ($\Delta^2$) of the quantum Tanner codes construction. A natural follow-up problem therefore is to look for ways to reduce the absolute runtime of the decoder, for example by reducing the check weights of the underlying code construction. 

Another problem is related to the decoding of the asymptotically good qLDPC codes by Panteleev and Kalachev~\cite{PK22}. While the quantum Tanner codes construction is in many ways similar to the codes by Panteleev and Kalachev, we do not currently know how to efficiently decode the Panteleev-Kalachev code. It would be interesting to see if our current decoder can be modified to work for the Panteleev-Kalachev code. A related -- and more generic -- problem is the existence of efficient decoders for good qLDPC codes constructed by the balanced product construction~\cite{BE21} in general, especially with the presence of non-trivial local codes.

Our current decoder requires the checking of local views belonging to vertices of both $V_0$ and $V_1$. This is in contrast to the small-set-flip decoder, which only requires checking the supports of generators of a single type. It may be possible that a tighter analysis (for example, using a stronger version of the low-overlap property, or more robust local codes) may allow us to eliminate the need to check both vertex types. A better understanding of the candidate flip-sets in general may be useful, especially towards the problem of lowering the runtime mentioned earlier.

\vspace{1 cm}
\textbf{Acknowledgements.} The authors would like to thank Zhiyang He for his valuable suggestions and his careful reading of an earlier draft of this work. We also thank Anand Natarajan, John Preskill, and Michael Beverland for helpful comments and discussions. S.G. acknowledges funding from the U.S. Department of Energy (DE-AC02-07CH11359), and the National Science Foundation (PHY-1733907). C.A.P. acknowledges funding from the Air Force Office of Scientific Research (AFOSR), FA9550-19-1-0360. The Institute for Quantum Information and Matter is an NSF Physics Frontiers Center. E.T. acknowledges funding received from DARPA 134371-5113608, DOD grant award KK2014, and the Center for Theoretical Physics at the Massachusetts Institute of Technology.

\bibliographystyle{unsrtnat}
\bibliography{Decoderbib}

\appendix

\section{Existence of dual tensor codes with sufficiently high robustness}
\label{sec:robustness}

In this appendix, we show the existence of dual tensor codes with sufficiently high robustness, which we require as a component of the quantum Tanner codes construction in order to prove correctness of our decoder. We will use the following notation throughout this section. Given codes $C_A$ and $C_B$ defined by parity check matrices $H_A$ and $H_B$, we denote their dual tensor code $(C_A^\perp \otimes C_B^\perp)^\perp$ by $C_{AB}$ for short, with the dependence on $C_A,C_B$ being implicit. 

We first recall the definition of a $w$-robust dual tensor code as defined in~\cite{LZ22}.

\begingroup
\def\thetheorem{\ref{def:robustness}}
\begin{definition}[$w$-Robustness]
Let $C_A,C_B\subseteq \bb F_2^n$ be classical codes with distances $d_A$ and $d_B$ respectively. We say that the dual tensor code $C_{AB}=C_A\otimes \bb F_2^n + \bb F_2^n\otimes C_B$ is \emph{$w$-robust} if every codeword $X\in C_{AB}$ with $|X|\le w$ is supported on the union of at most $|X|/d_A$ non-zero columns and $|X|/d_B$ non-zero rows. That is, there exist rows $A'$ with $|A'|\ge n-|X|/d_B$ and columns $B'$ with $|B'|\ge n-|X|/d_A$ such that $\res{X}{A'\times B'}=0$.
\end{definition}
\addtocounter{theorem}{-1}
\endgroup

\begin{definition}[Sufficiently Robust]
We say that $C_{AB}$ is \emph{sufficiently robust} if there exists some $\varepsilon > 0$ such that $C_{AB}$ is $\Delta^{3/2+\varepsilon}$-robust.
\end{definition}

When a codeword of a dual tensor code is supported on few columns and rows, it has a decomposition into column and row codewords respecting this support.

\begin{lemma}\label{lem:rowcolsum}
Let $C_A$ and $C_B$ be classical codes of distance at least $d$ and $C=C_A\otimes \mathbb F_2^B + \mathbb F_2^A\otimes C_B$ be the dual tensor code. Suppose $X\in C$ is supported on the union of $\alpha$ non-zero rows and $\beta$ non-zero columns, with $\alpha,\beta < d$. Then $X$ can be written as $X=\mathbf r + \mathbf c$ where $\mathbf r\in \mathbb F_2^A\otimes C_B$ is supported on at most $\alpha$ non-zero rows and $\mathbf c\in C_A\otimes \mathbb F_2^B$ is supported on at most $\beta$ non-zero columns.
\end{lemma}

\begin{proof}
Let $\overline{A'}$ be the rows and $\overline{B'}$ be the columns that $X$ is supported on. We have $\alpha = |\overline{A'}|$ and $\beta = |\overline{B'}|$. Let $C_{A'},C_{B'}$ be the projections of $C_A$ and $C_B$ onto the complements $A'$ and $B'$ respectively. Because $|\overline{A'}|,|\overline{B'}|<d$, the projections $C_A\to C_{A'}$ and $C_B\to C_{B'}$ are isomorphisms, and hence so is the projection $C_A\otimes C_B\to C_{A'}\otimes C_{B'}$.

Let $X=\mathbf r + \mathbf c$ be any decomposition where $\mathbf r\in \mathbb F_2^A\otimes C_B$ and $\mathbf c\in C_A\otimes \mathbb F_2^B$. By assumption, we have $\left. X\right|_{A'\times B'} = 0$, so we have that $\left.\mathbf r\right|_{A'\times B'}=\left.\mathbf c\right|_{A'\times B'}$. It follows that this quantity is in $C_{A'}\otimes C_{B'}$. By the isomorphism above, there exists a unique $Y\in C_A\otimes C_B$ such that $\left.Y\right|_{A'\times B'}=\left.\mathbf r\right|_{A'\times B'}=\left.\mathbf c\right|_{A'\times B'}$. Again due to the isomorphism above, we actually know that $Y$ is equal to $\mathbf r$ on the rows indexed by $A'$, and also that $Y$ is equal to $\mathbf c$ on the columns indexed by $B'$. Therefore, $X=(\mathbf r+Y) + (\mathbf c+Y)$ is the desired decomposition with $\left. (\mathbf r+Y) \right|_{A'\times B} = 0$ and $\left. (\mathbf c+Y) \right|_{A\times B'} = 0$.
\end{proof}

We use a probabilistic argument to show that randomly chosen dual tensor codes will be sufficiently robust with high probability. There are several ways to randomly choose a classical code, which we make use of in different parts of the proof. We first show that these distributions are almost the same.

\subsection{Lemmas about random codes}

In this subsection, we collect some basic results about various ensembles of random codes. The main utility of these results is in the proof of Theorem~\ref{thm:puncture_robust}, where we must consider random ensembles of punctured codes. While the majority of the results in this appendix are more conveniently shown using ensembles of codes obtained from random parity check matrices, it is much simpler to perform puncturing on codes defined using generator matrices. The results proven in this subsection will allow us to freely switch between the various closely related ensembles of random codes so that we may use the most convenient ensemble at each step.

Let $C_1,C_2,C_3$ be random classical codes of length $\Delta$ chosen from three different ensembles:
\begin{enumerate}
    \item Let $H\sim \mathcal U\left(\bb F_2^{(1-\rho)\Delta\times\Delta}\right)$ be a uniformly random parity check matrix and let $C_1=\ker H$.
    \item Let $G\sim \mathcal U\left(\bb F_2^{\Delta\times\rho\Delta}\right)$ be a uniformly random generator matrix and let $C_2=\operatorname{col} H$.
    \item Let $\mc S=\{C\subseteq\bb F_2^\Delta: C \text{ is a $\rho\Delta$-dimensional subspace}\}$ and let $C_3\sim\mc U(\mc S)$ be a uniformly random $\rho\Delta$-dimensional subspace.
\end{enumerate}

\begin{lemma}
For a fixed $C\in \mc S$, we have
\begin{equation}
    \Pr(C_1=C \mid \rk H=(1-\rho)\Delta) = \Pr(C_2=C \mid \rk G=\rho\Delta) = \Pr(C_3=C)\, .
\end{equation}
\end{lemma}

\begin{proof}
We first prove that $\Pr(C_1=C \mid \rk H=(1-\rho)\Delta)=\Pr(C_3=C)$. Since $C_3$ is drawn from a uniform distribution, it is sufficient to show that given two $\rho\Delta$-dimensional subspaces $C',C'' \in \mc S$, we have 
\begin{align}
\Pr(C_1=C' \mid \rk H=(1-\rho)\Delta)=\Pr(C_1=C'' \mid \rk H=(1-\rho)\Delta)\, .
\end{align}
Equivalently, we show that the number of full rank matrices $H$ with $\ker H=C'$ is the same as the number with $\ker H=C''$. Let
\begin{align}
    \mc H_1 &= \{H\in \bb F_2^{(1-\rho)\Delta\times\Delta}: \rk H=(1-\rho)\Delta,\, \ker H=C'\}\, ,\\
    \mc H_2 &= \{H\in \bb F_2^{(1-\rho)\Delta\times\Delta}: \rk H=(1-\rho)\Delta,\, \ker H=C''\}\, .
\end{align}
Because $C'$ and $C''$ have the same dimension, there is an invertible matrix $A\in \bb F_2^{\Delta\times\Delta}$ such that $AC'=C''$. Consider the bijective linear map
\begin{align}
    f: \bb F_2^{(1-\rho)\Delta\times\Delta} &\to \bb F_2^{(1-\rho)\Delta\times\Delta}\, ,\\
    H &\mapsto HA^{-1}\, .
\end{align}
Now if $\ker H=C'$, then for any $x\in C''$, we have
\begin{equation}
    f(H)x = HA^{-1}x = 0
\end{equation}
since $A^{-1}x\in C'$. Thus, $f$ restricts to a bijection between $\mc H_1$ and $\mc H_2$.

The other equality is shown similarly. We prove that the two sets
\begin{align}
    \mc G_1 &= \{G\in \bb F_2^{\Delta\times\rho\Delta}: \rk G=\rho\Delta,\, \col G=C'\} \\
    \mc G_2 &= \{G\in \bb F_2^{\Delta\times\rho\Delta}: \rk G=\rho\Delta,\, \col G=C''\}
\end{align}
have the same cardinality. Define
\begin{align}
    g: \bb F_2^{\Delta\times\rho\Delta} &\to \bb F_2^{\Delta\times\rho\Delta}\, ,\\
    G &\mapsto AG\, .
\end{align}
Suppose $G\in \mc G_1$. For any $y\in C_2$, we have $A^{-1}y\in C_1$, so let $x\in \bb F_2^{\rho\Delta}$ be such that $Gx=A^{-1}y$. Then
\begin{equation}
    g(G)x = AGx = AA^{-1}y = y\, .
\end{equation}
This shows that $g(G)\in \mc G_2$, and so $g$ is a bijection between $\mc G_1$ and $\mc G_2$.
\end{proof}

\begin{lemma}\label{lem:prob_not_fullrank}
The probability that $H$ or $G$ is not full rank is exponentially small:
\begin{equation}
    \Pr(\rk H\ne (1-\rho)\Delta)\le 2^{-\rho\Delta} \quad \text{ and } \quad \Pr(\rk G\ne \rho\Delta)\le 2^{-(1-\rho)\Delta}\, .
\end{equation}
\end{lemma}

\begin{proof}
Let the columns of $G$ be $g_1,g_2,\dots,g_{\rho\Delta}$. If $G$ is not full rank, there must be a non-trivial subset of the columns that sums to zero. Thus, a union bound gives
\begin{align}
    \Pr(\rk G\ne \rho\Delta) &= \Pr\left(\sum_{i\in S} g_i=0 \text{ for some nonempty subset } S\subseteq [\rho\Delta]\right)\\
    &\le \sum_{\emptyset\ne S\subseteq [\rho\Delta]} \Pr\left(\sum_{i\in S} g_i=0\right)\\
    &\le 2^{\rho\Delta}2^{-\Delta}\\
    &= 2^{-(1-\rho)\Delta}\,  .
\end{align}
The same argument shows that $\Pr(\rk H\ne (1-\rho)\Delta)\le 2^{-\rho\Delta}$.
\end{proof}

The above two lemmas imply that statements about random codes do not depend much on which distribution the codes are chosen from.

\begin{corollary}\label{cor:total_var_bound}
Let $\mc V$ denote the set of all subspaces of $\mathbb{F}^\Delta_2$. Then the total variation distance $\delta_{TV}(\Pr_{C_i},\Pr_{C_j})$ between the distributions of $C_i$ and $C_j$ is bounded above by
\begin{align}
    \delta_{TV}(\mathrm{Pr}_{C_i},\mathrm{Pr}_{C_j}) \equiv \frac{1}{2}\sum_{C \in \mc V}\big|\Pr_{C_i}(C_i = C) - \Pr_{C_j}(C_j = C)\big| \le 2^{-\Omega(\Delta)}
\end{align}
for $i,j \in \{1,2,3\}$.
\end{corollary}

\begin{proof}
Let us compare the distributions of $C_1$ and $C_3$. Note that $C_3$ is uniformly random on $\mc S$ and zero on $\mc V\backslash \mc S$. Therefore we can write
\begin{align}
    \delta_{TV}(\mathrm{Pr}_{C_1},\mathrm{Pr}_{C_3}) &= \frac{1}{2}\sum_{C\in\mc V}\big|\Pr_{C_1}(C_1 = C) - \Pr_{C_3}(C_3 = C)\big|\\
    &= \frac{1}{2}\sum_{C\in\mc S}\big|\Pr_{C_1}(C_1 = C) - \Pr_{C_3}(C_3 = C)\big| + \frac{1}{2}\sum_{C \in \mc V\backslash \mc S}\Pr_{C_1}(C_1 = C)\\
    &= \frac{1}{2}\sum_{C\in\mc S}\big|\Pr_{C_1}(C_1 = C) - \Pr_{C_3}(C_3 = C)\big| + \frac{1}{2}\Pr_{C_1}(\dim C_1 \neq \rho\Delta)\\
    &\le \frac{1}{2}\sum_{C\in\mc S}\big|\Pr_{C_1}(C_1 = C) - \Pr_{C_3}(C_3 = C)\big| + \frac{1}{2}\cdot 2^{-\rho\Delta}\, ,
\end{align}
where the last inequality follows by Lemma~\ref{lem:prob_not_fullrank}. For $C\in \mc S$, the previous two lemmas imply that
\begin{align}
    \Pr_{C_1}(C_1=C) &= \Pr_{H}(C_1=C \mid \rk H=(1-\rho)\Delta)\cdot\Pr_{H}(\rk H=(1-\rho)\Delta)\\
    &\ge \Pr_{C_3}(C_3=C)(1-2^{-\rho\Delta})\, .
\end{align}
It follows that
\begin{align}
    \sum_{C\in\mc S}\big|\Pr_{C_1}(C_1 = C) - \Pr_{C_3}(C_3 = C)\big| \le \sum_{C\in\mc S}\Pr_{C_3}(C_3 = C)2^{-\rho\Delta} = 2^{-\rho\Delta}\, .
\end{align}
It follows that we have
\begin{align}
    \delta_{TV}(\mathrm{Pr}_{C_1},\mathrm{Pr}_{C_3}) \le 2^{-\rho\Delta}\, .
\end{align}
The same argument holds when comparing $C_2$ and $C_3$, with the upper bound $2^{-(1-\rho)\Delta}$.
\end{proof}

Note that the total variation distance can equivalently be given by
\begin{align}
     \delta_{TV}(\mathrm{Pr}_{C_i},\mathrm{Pr}_{C_j}) = \sup_{A \subseteq \mc V}\big|\Pr_{C_i}(C_i \in A) - \Pr_{C_j}(C_j \in A)\big|\, .
\end{align}
The most common way we will apply Corollary~\ref{cor:total_var_bound} is in terms of joint probability distributions. For independent random variables, the total variation distance satisfies
\begin{align}
    \delta_{TV}(\mathrm{Pr}_{C_i,C_j}, \mathrm{Pr}_{C_i,C_k}) \le \delta_{TV}(\mathrm{Pr}_{C_j}, \mathrm{Pr}_{C_k})\, .
\end{align}
This allows us to freely switch between the various joint distributions, up to an exponentially small overhead.

\subsection{Random codes are sufficiently robust}
Throughout this section, we will use the notation $\tilde{\Theta}(f(x))$ to denote $\Theta(f(x)\log f(x))$. 
For $a \in (0,1)$, we have the following asymptotic bound for the binomial coefficients which we will use frequently:
\begin{equation}
    \binom{n}{n^a} = 2^{\tilde{\Theta}(n^a)}\label{eq:binom_id}\, .
\end{equation}
Note that equation~\eqref{eq:binom_id} follows from the bound~\cite{TC06}
\begin{align}
\frac{1}{n+1} 2^{n h(k/n)} \le \binom{n}{k} \le 2^{n h(k/n)}
\end{align}
after some basic algebra. Here, $h(x)$ denotes the binary entropy function.

The goal of this section is to show that randomly chosen dual tensor codes will be sufficiently robust with high probability. Towards this goal, it will be more convenient to work with a condition which is proxy for $w$-robustness, one which we will call \emph{sparse robustness} (and its associated punctured version). In all that follows we will fix some small but otherwise arbitrary constant $\varepsilon > 0$. All definitions below are technically made with reference to some chosen $\varepsilon$, but we will suppress the dependence out of brevity. 

\begin{definition}[Low-Weight and Sparse]
We will say that a matrix $X \in \mathbb{F}_2^{\Delta\times\Delta}$ is \emph{low-weight} if $|X| \le \Delta^{3/2+\varepsilon}$. We will say that $X$ is \emph{sparse} if each row and column of $X$ has weight at most $\Delta^{1/2+2\varepsilon}$.
\end{definition}

Note that low-weight and sparse above are closely related but distinct notions. Neither implies the other. We ultimately want to show robustness against low-weight codewords, and we do so by first showing robustness against sparse matrices.

\begin{definition}[Sparse Robustness and Puncturing]
Let $C_{AB}$ be a dual tensor code with distance $d\ge\delta\Delta$. We say that $C_{AB}$ is \emph{sparse robust} if $C_{AB}$ does not contain any non-zero sparse codewords.

Let $C_A \subseteq \mathbb{F}_2^A$ be a code and let $A' \subseteq A$. We say that the code $C_{A'} \subseteq \mathbb{F}_2^{A'}$ is a \emph{punctured code} obtained from $C_A$ if the codewords of $C_{A'}$ are precisely those obtained from $C_A$ by removing all entries in $\overline{A'} = A\backslash A'$. In this case, we also say that $C_{A'}$ is obtained from $C_A$ by puncturing on $\overline{A'}$. Note that a generator matrix for $C_{A'}$ is obtained from a generator matrix for $C_A$ by removing the entries supported on $\overline{A'}$.

Let $\mathcal{P}$ denote the set of all codes $C_{A'B'}$ obtained from $C_{AB}$ by puncturing $A$ and $B$ on $\Delta^{1-\varepsilon}$ coordinates (note that $|A'|=|B'| = \Delta-\Delta^{1-\varepsilon}$ in this case). We say that $C_{AB}$ is \emph{sparse robust with respect to puncturing (SRP)} if every $C_{A'B'} \in \mathcal{P}$ is sparse robust.
\end{definition}

The connection between $w$-robustness and sparse robustness is formalized in the lemma below.

\begin{lemma}\label{lem:sparse_implies_robust}
Let $C_{AB}$ be a dual tensor code with distance $d=\delta\Delta$. For sufficiently large $\Delta$, if $C_{AB}$ is sparse robust with respect to puncturing then it is $\Delta^{3/2+\varepsilon/2}$-robust. In particular, $C_{AB}$ is sufficiently robust.
\end{lemma}
\begin{proof}

Let $C_{AB}$ be sparse robust with respect to puncturing. Let $X \in C_{AB}$ be a codeword of weight $|X| \le \Delta^{3/2+\varepsilon/2}$. From Lemma~30 of~\cite{LZ22}, if $X$ is supported on the union of at most $d/2=\delta\Delta/2$ rows and columns, then it is supported on the union of $|X|/d$ rows and columns. Therefore, it suffices to show that $X$ is supported on the union of at most $\delta\Delta/2$ non-zero rows and columns.

Since $|X|\le \Delta^{3/2+\varepsilon/2}$, it can have at most $\Delta^{1-\varepsilon}$ rows or columns which are of weight greater than $\Delta^{1/2+3\varepsilon/2}$. By removing these high-weight rows and columns, it follows that there exists some puncturing sets $\overline{A'}, \overline{B'}$ of size $\Delta^{1-\varepsilon}$ such that $X$ punctured on those coordinates has all columns and rows with weight at most $\Delta^{1/2+3\varepsilon/2}$. The idea now is to show that the punctured matrix $X'$ is sparse, so that it must vanish by the sparse robustness of the punctured code $C_{A'B'}$. The sparsity of $X'$ is slightly complicated by the fact that it is a matrix of size $\Delta' = \Delta - \Delta^{1-\varepsilon} < \Delta$. To account for the smaller size of $\Delta'$, let us choose $\Delta$ to be sufficiently large so that $\Delta' \ge \Delta/2$. Then we have
\begin{align}
\Delta^{1/2+3\varepsilon/2} \le \left(2\Delta'\right)^{1/2+3\varepsilon/2} \le (\Delta')^{1/2+2\varepsilon}\, ,
\end{align}
where the last inequality holds as long as we choose $\Delta$ large enough so that $4 \le \varepsilon\log_2\Delta$. It follows that for sufficiently large $\Delta$, the rows and columns of $X'$ have at most $(\Delta')^{1/2+2\varepsilon}$ entries, so the punctured matrix $X'$ is sparse. Since $C_{A'B'}$ is sparse robust by assumption, it follows that $X'=0$. Therefore $X$ must have been supported on its punctured rows and columns, of which there are $O(\Delta^{1-\varepsilon})$. This will be less than $d/2 = \delta\Delta/2$ for sufficiently large $\Delta$, and the result follows.
\end{proof}

We will therefore proceed by first showing that a randomly chosen dual tensor code will be sparse robust with high probability, and then use this fact to show that random dual tensor codes are sparse robust with respect to puncturing---and hence sufficiently robust---with high probability.

For a dual tensor code $C_{AB}$ with distance $d\ge\delta\Delta$, we are automatically guaranteed that there are no non-zero sparse codewords supported on fewer than $\delta\Delta$ non-zero rows and columns.

\begin{lemma}\label{lem:sparse_row_col}
Let $C_{AB}$ be a dual tensor code with distance $d\ge \delta\Delta$. Let $\gamma \in (0,1)$ be some constant. For $\Delta$ sufficiently large, the dual tensor code $C_{AB}$ contains no non-zero sparse codewords which are supported on the union of $\le \gamma\delta\Delta$ non-zero rows and $\le \gamma\delta\Delta$ non-zero columns.
\end{lemma}
\begin{proof}
Suppose that $X\in C_{AB}$ is sparse and is supported on a union of at most $\delta\Delta$ non-zero rows and columns. By Lemma~\ref{lem:rowcolsum}, there exists a decomposition $X = \mathbf{r}+\mathbf{c}$ where $\mathbf{c}\in C_A\otimes \mathbb{F}_2^B$ and $\mathbf{r} \in \mathbb{F}_2^A\otimes C_B$ such that $\mathbf{c}$ has $\le \gamma\delta\Delta$ non-zero columns, each of which is a codeword for $C_A$, and $\mathbf{r}$ has $\le\gamma\delta\Delta$ non-zero rows, each of which is a codeword for $C_B$. Since $X$ is sparse, it follows that each column of $\mathbf{c}$ has weight
\begin{align}
    |\mathbf{c}[\cdot, i]| \le  \gamma\delta\Delta + \Delta^{1/2+2\varepsilon}\, .
\end{align}
Choosing $\Delta$ large enough so that $\Delta^{1/2+2\varepsilon} < (1-\gamma)\delta\Delta$, we get
\begin{align}
    |\mathbf{c}[\cdot, i]| < \delta\Delta
\end{align}
so that $\mathbf{c}[\cdot,i] = 0$. Since this holds for every column, it follows that $\mathbf{c}$ is the zero codeword. The same logic applies to $\mathbf{r}$. 
\end{proof}

It follows that to show a random $C_{AB}$ is sparse robust, it suffices to show that it cannot contain any sparse codewords with more than $\delta\Delta/2$ non-zero columns and more than $\delta\Delta/2$ non-zero rows

\begin{theorem}[Sparse Robustness]\label{thm:sparse_robust}
Fix constants\footnote{Note that we are only interested in the upper bound on the probability here, so we do not restrict the values of the $\rho_A$, $\rho_B$, and $\delta$. With particular choices of $\rho_A$, $\rho_B$, and $\delta$, this result implies that sparse robust codes exist by the Gilbert–Varshamov Bound.} $\rho_A,\rho_B \in (0,1)$, $\varepsilon \in (0,1/14)$, and $\delta \in (0,1)$.

Let $H_A\in \mathbb F_2^{(1-\rho_A)\Delta\times\Delta}$ and $H_B\in \mathbb F_2^{(1-\rho_B)\Delta\times\Delta}$ be uniformly random binary check matrices defining codes $C_A$ and $C_B$ respectively. Then the probability that $C_{AB}$ has distance $d\ge \delta\Delta$ and is not sparse robust is bounded above by
\begin{align}
    \Pr_{H_A,H_B}(C_{AB} \text{ is not SR} \text{ and } d \ge \delta\Delta) \le 2^{-\Theta(\Delta^{3/2-2\varepsilon})}\, .
\end{align}
\end{theorem}

To prove Theorem~\ref{thm:sparse_robust} we first begin with some setup. Let us define $\mathcal{X}\subseteq \mathbb{F}_2^{\Delta\times\Delta}$ as the set of all sparse matrices with more than $\delta\Delta/2$ non-zero rows and columns, i.e.,
\begin{align}
    \mathcal{X} = \{X \in \mathbb{F}^{\Delta\times\Delta}_2 \mid X \text{ is sparse and has}>\delta\Delta/2 \text{ non-zero rows and}>\delta\Delta/2 \text{ non-zero columns}\}\, .
\end{align}

We first bound the number of high and low rank matrices in $\mathcal{X}$.

\begin{lemma}\label{lem:lowrank_card_bound}
Let $b\in(0, 1/2)$ and let
\begin{align}
    \mathcal{X}_1 = \{X \in \mathcal{X} \mid \mathrm{rank}(X) \le \Delta^{1/2+b}\}\quad\text{and}\quad \mathcal{X}_2 = \mathcal{X}\backslash \mathcal{X}_1 = \{X \in \mathcal{X} \mid \mathrm{rank}(X) > \Delta^{1/2+b}\}\, .
\end{align}
Then we have the cardinality bounds\footnote{Note that these bounds only make use of the sparsity condition, and not the restriction on the number of rows and columns.}
\begin{align}
    |\mathcal{X}_1| \le 2^{\tilde{\Theta}(\Delta^{1+2\varepsilon+b})} \quad\text{and}\quad |\mathcal{X}_2|\le 2^{\tilde{\Theta}(\Delta^{3/2+2\varepsilon})}\, .
\end{align}
\end{lemma}
\begin{proof}
We begin with the proof of the high rank case. Since the overwhelming majority of matrices in $\mathcal{X}$ are expected to be high rank, we simply bound the total number of matrices in $\mathcal{X}$ as a whole. Since each matrix in $\mathcal{X}$ is sparse, it can have weight at most $\Delta^{3/2+2\varepsilon}$. We can therefore bound $|\mathcal{X}|$ by the total number of matrices of such weight, given by
\begin{align}
    |\mathcal{X}| \le \sum_{j=0}^{\Delta^{3/2+2\varepsilon}}\binom{\Delta^2}{j} \le \Delta^{3/2+2\varepsilon}\binom{\Delta^2}{\Delta^{3/2+2\varepsilon}} = 2^{\tilde{\Theta}(\Delta^{3/2+2\varepsilon})}\, .
\end{align}

Now we bound the low rank case. Let us see how many ways we can build some $X \in \mathcal{X}_1$ with $\rk X = N$. We first fix a basis for the row space of $X$. Since $X$ is sparse, each basis vector can be chosen in at most
\begin{align}
    \sum_{j=1}^{\Delta^{1/2+2\varepsilon}}\binom{\Delta}{j} \le \Delta^{1/2+2\varepsilon}\binom{\Delta}{\Delta^{1/2+2\varepsilon}} = 2^{\tilde{\Theta}(\Delta^{1/2+2\varepsilon})}
\end{align}
ways. There are $N$ basis vectors, so there are at most
\begin{align}
\left(2^{\tilde{\Theta}(\Delta^{1/2+2\varepsilon})}\right)^{N} = 2^{\tilde{\Theta}(N\Delta^{1/2+2\varepsilon})}
\end{align}
possible (ordered) bases for the row space of $X$. We can place these basis vectors into the rows of the matrix $X$ in at most
\begin{align}
    \binom{\Delta}{N} \le \binom{\Delta}{\Delta^{1/2+b}} = 2^{\tilde{\Theta}(\Delta^{1/2+b})}
\end{align}
ways. Having fixed a row space basis, each of the remaining rows must be a linear combination of these basis vectors. By row reduction, let $\{v_1,\dots,v_N\}$ be another basis for the row space of $X$ such that each $v_i$ has a $1$ in some column $c_i$ in which every other $v_j$ is $0$. Now, every row of $X$ is also a linear combination of $\{v_1,\dots,v_N\}$. If the basis vector $v_i$ appears in the linear combination defining a row $r_j$, then $(r_j)_{c_i}=1$. However, by column sparsity, $(r_j)_{c_i}=1$ can only be true for at most $\Delta^{1/2+2\varepsilon}$ values of $j$. There are therefore at most 
\begin{align}
    \sum_{j=0}^{\Delta^{1/2+2\varepsilon}}\binom{\Delta}{j} \le \Delta^{1/2+2\varepsilon} \binom{\Delta}{\Delta^{1/2+2\varepsilon}} = 2^{\tilde{\Theta}(\Delta^{1/2+2\varepsilon})}
\end{align}
ways to choose the rows which contain a given $v_i$ in its linear combination. Making this choice for each $v_i$, it follows that there are at most
\begin{align}
    \left(2^{\tilde{\Theta}(\Delta^{1/2+2\varepsilon})}\right)^{N} = 2^{\tilde{\Theta}(N\Delta^{1/2+2\varepsilon})}
\end{align}
ways to fill out the remaining rows of the matrix, since we have chosen the subset of $\{v_1,\dots,v_N\}$ in the linear combination defining every row of $X$. Combining everything, and summing over the rank $N$, it follows that there can be at most
\begin{align}
    \sum_{N=1}^{\Delta^{1/2+b}}2^{\tilde{\Theta}(N\Delta^{1/2+2\varepsilon})}2^{\tilde{\Theta}(\Delta^{1/2+b})}2^{\tilde{\Theta}(N\Delta^{1/2+2\varepsilon})} \le \Delta^{1/2+b}\cdot 2^{\tilde{\Theta}(\Delta^{1+2\varepsilon+b})} = 2^{\tilde{\Theta}(\Delta^{1+2\varepsilon+b})}
\end{align}
distinct matrices in $\mathcal{X}_1$.
\end{proof}

For low-rank matrices $X\in \mc X_1$, we want to show that $H_AX$ is also likely to have low rank. This uses the following lemma:

\begin{lemma}\label{lem:rank_prob}
Let $Y\in \bb F_2^{\Delta\times \Delta'}$ be a matrix of rank $M$ and let $H\in \bb F_2^{(1-\rho)\Delta\times\Delta}$ be chosen uniformly at random. Then for any $K$,
\begin{equation}
    \Pr_H(\rk(HY)=K)\le \binom{M}{K}2^{-((1-\rho)\Delta-K)(M-K)}\, .
\end{equation}
\end{lemma}
\begin{proof}
Let $y_1,\dots,y_M$ be linearly independent columns of $Y$. If $\rk(HY)=K$, then $\{Hy_1,\dots,Hy_M\}$ must span a $K$-dimensional subspace of $\bb F_2^{(1-\rho)\Delta}$. In other words, there is a $K$-element subset $S\subseteq [M]$ such that $V_S\equiv \spn\{Hy_j\}_{j\in S}$ is $K$-dimensional and $Hy_i\in V_S$ for all $i\in [M]$. Let $\mc S$ denote the set of all $K$-element subsets of $[M]$. We have
\begin{align}
    &\Pr_H(\rk(HY)=K)\\ 
    = &\Pr_H(\exists S\in\mc S \text{ such that } \dim V_S=K \text{ and } Hy_i\in V_S \text{ for all } i\in [M])\\
    \le &\sum_{S\in\mc S} \Pr_H(\dim V_S=K \text{ and } Hy_i\in V_S \text{ for all } i\in [M])\\
    \le &\sum_{S\in\mc S} \Pr_H(Hy_i\in V_S \text{ for all } i\in [M] \mid \dim V_S = K)\\
    = &\sum_{S\in \mc S} \prod_{i \in [M]} \Pr_H(Hy_i\in V_S \mid \dim V_S = K)\, .
\end{align}
Now, consider some fixed $S$ in the latter sum. If $i\in S$, then $Hy_i\in V_S$ is guaranteed. Otherwise, because the $y_i$ are independent, the $Hy_i$ are independently mapped to uniformly random vectors in $\bb F_2^{(1-\rho)\Delta}$, and each of them lands in the fixed subspace $V_S$ with probability $2^{-((1-\rho)\Delta-K)}$. Thus,
\begin{align}
    \Pr_H(\rk(HY)=K)\le \binom{M}{K}\left(2^{-((1-\rho)\Delta-K)}\right)^{M-K}\, .
\end{align}
\end{proof}

We will also need the following fact about the rank of sparse matrices, which first appears in~\cite{LZ22}.

\begin{lemma}[Corollary 25 of~\cite{LZ22}]\label{lem:LZ25}
Let $C_A$ be an error-correcting code with minimum distance $d_A\ge\delta\Delta$. Let $H_A$ be its parity check matrix. Let $X \in \mathbb{F}^{\Delta\times \Delta}_2$ be a matrix such that all columns are of weight at most $\Delta^{1/2+2\varepsilon}$, and such that $X$ has more than $\delta\Delta/2$ non-zero rows. Then for $\Delta$ sufficiently large, we have $\mathrm{rank}(H_AX) \ge (\delta/2)\Delta^{1/2-2\varepsilon}$.
\end{lemma}
\begin{proof}
This follows directly from the proofs of Lemma 24 and Corollary 25 in~\cite{LZ22} with the appropriate modifications of the relevant parameters.
\end{proof}

Now we are ready to prove Theorem~\ref{thm:sparse_robust}.

\begin{proof}[Proof of Theorem~\ref{thm:sparse_robust}]
It follows from Lemma~\ref{lem:sparse_row_col} that a dual tensor code $C_{AB}$ with distance $d\ge \delta\Delta$ is sparse robust if and only if it contains no element of $\mathcal{X}$. Taking a union bound over $\mathcal{X}$, we can write
\begin{align}
    \Pr_{H_A,H_B}(C_{AB} \text{ is not SR} \text{ and } d \ge \delta\Delta) &\le \sum_{X\in\mathcal{X}}\Pr_{H_A,H_B}(X\in C_{AB} \text{ and }d\ge\delta\Delta)\\
    &= \sum_{X\in\mathcal{X}}\Pr_{H_A,H_B}(H_AXH_B^\mathrm{T} = 0  \text{ and }d\ge\delta\Delta)\, ,
\end{align}
where the last line follows from the definition of the dual tensor code. To proceed, we decompose the sum according to the rank of $H_AX$. It follows from Lemma~\ref{lem:rank_prob} (with $K=0$) that for any matrix $Y$ with $\mathrm{rank}(Y) = M$, the probability over $H_B$ that $YH_B^\mathrm{T} = 0$ is bounded above by $2^{-(1-\rho_B)\Delta M}$. Applying this fact by taking $Y = H_AX$, we get
\begin{align}
    &\sum_{X\in\mathcal{X}}\Pr_{H_A,H_B}(H_AXH_B^\mathrm{T} = 0 \text{ and }d\ge\delta\Delta)\\ 
    \le & \sum_{X\in\mathcal{X}}\Pr_{H_A,H_B}(H_AXH_B^\mathrm{T} = 0 \text{ and }d_A\ge\delta\Delta)\\
    = &\sum_{X\in\mathcal{X}}\sum_{M=0}^\mathrm{rank(X)}\bigg[\Pr_{H_A,H_B}(H_AXH_B^\mathrm{T} = 0 \mid \mathrm{rank}(H_AX)=M \text{ and }d_A\ge\delta\Delta)\nonumber\\
    &\qquad\qquad\qquad \times\Pr_{H_A,H_B}(\mathrm{rank}(H_AX)=M \text{ and }d_A\ge\delta\Delta)\bigg]\\
    \le &\sum_{X\in\mathcal{X}}\sum_{M=0}^{\mathrm{rank}(X)}2^{-(1-\rho_B)\Delta M}\Pr_{H_A}(\mathrm{rank}(H_AX) = M \text{ and }d_A\ge\delta\Delta)\, .
\end{align}

We can bound the inner sum using Lemma~\ref{lem:LZ25}.  Note that any $X \in \mathcal{X}$ satisfies the hypotheses of Lemma~\ref{lem:LZ25}. Therefore we get
\begin{align}
    &\quad\sum_{M=0}^{\mathrm{rank}(X)}\Pr_{H_A}(\mathrm{rank}(H_AX) = M\text{ and }d_A\ge\delta\Delta)2^{-(1-\rho_B)\Delta M}\\
    = &\sum_{M=(\delta/2)\Delta^{1/2-2\varepsilon}}^{\mathrm{rank}(X)}\Pr_{H_A}(\mathrm{rank}(H_AX) = M \text{ and }d_A \ge \delta\Delta )2^{-(1-\rho_B)\Delta M}\\
    \le &\sum_{M=(\delta/2)\Delta^{1/2-2\varepsilon}}^{\mathrm{rank}(X)}\Pr_{H_A}(\mathrm{rank}(H_AX) = M)2^{-(1-\rho_B)\Delta M}\, ,
\end{align}
where we drop the distance condition in the last line since it has now played its part in allowing the application of Lemma~\ref{lem:LZ25}.

We will now bound the total probability in two stages by splitting the outer sum (see Lemma~\ref{lem:lowrank_card_bound}) into a low rank part $\mathcal{X}_1\subseteq \mathcal{X}$ (where $\rk X \le \Delta^{1/2+b}$) and a high rank part $\mathcal{X}_2 \subseteq \mathcal{X}$ (where $\rk X > \Delta^{1/2+b}$), with $b \in (2\varepsilon, 3\varepsilon)$, to get
\begin{align}
&\sum_{X\in\mathcal{X}}\sum_{M=(\delta/2)\Delta^{1/2-2\varepsilon}}^{\mathrm{rank}(X)}\Pr_{H_A}(\mathrm{rank}(H_AX) = M)2^{-(1-\rho_B)\Delta M}\\
=&\underbrace{\sum_{X\in\mathcal{X}_1}\sum_{M=(\delta/2)\Delta^{1/2-2\varepsilon}}^{\mathrm{rank}(X)}\Pr_{H_A}(\mathrm{rank}(H_AX) = M)2^{-(1-\rho_B)\Delta M}}_{\equiv P_1}\\ 
+ &\underbrace{\sum_{X\in\mathcal{X}_2}\sum_{M=(\delta/2)\Delta^{1/2-2\varepsilon}}^{\mathrm{rank}(X)}\Pr_{H_A}(\mathrm{rank}(H_AX) = M)2^{-(1-\rho_B)\Delta M}}_{\equiv P_2}\, .
\end{align}

\begin{proof}[Bound for $P_1$]
We can bound the low rank part $P_1$ using the cardinality bound for $|\mathcal{X}_1|$ in Lemma~\ref{lem:lowrank_card_bound}. We get
\begin{align}
    P_1 &= \sum_{X\in\mathcal{X}_1}\sum_{M=(\delta/2)\Delta^{1/2-2\varepsilon}}^{\mathrm{rank}(X)}\Pr_{H_A}(\mathrm{rank}(H_AX) = M)2^{-(1-\rho_B)\Delta M}\\    
    &\le \sum_{X\in\mathcal{X}_1} 2^{-(1-\rho_B)\Delta\cdot (\delta/2)\Delta^{1/2-2\varepsilon}}\\
    &= |\mathcal{X}_1|\cdot 2^{-\Theta(\Delta^{3/2-2\varepsilon})}\\
    &\le 2^{\tilde{\Theta}(\Delta^{1+2\varepsilon+b})}2^{-\Theta(\Delta^{3/2-2\varepsilon})}\\
    &= 2^{-\Theta(\Delta^{3/2-2\varepsilon})}\, .
\end{align}
The second line follows by bounding the inner sum using its largest term. The last line follows due to the fact that $4\varepsilon + b < 7\varepsilon < 1/2$, so that $\Delta^{3/2-2\varepsilon}$ asymptotically dominates $\Delta^{1+2\varepsilon+b}$.
\end{proof}

\begin{proof}[Bound for $P_2$]
To bound the expression $P_2$, we will assume without loss of generality that $\rho_B \le \rho_A$. If this is not the case, we can switch the roles of $C_A$ and $C_B$ by applying the current argument to the transposed code $C_{BA}$, noting that the set $\mathcal{X}$ is invariant under transpose. Writing $N = \mathrm{rank}(X)$, we can bound the inner sum of $P_2$ as
\begin{align}
    &\sum_{M=(\delta/2)\Delta^{1/2-2\varepsilon}}^{\mathrm{rank}(X)}\Pr_{H_A}(\mathrm{rank}(H_AX) = M)2^{-(1-\rho_B)\Delta M}\\ \le& \sum_{M=0}^{\min(N,(1-\rho_A)\Delta)}\binom{N}{M}2^{-((1-\rho_A)\Delta-M)(N-M)}2^{-(1-\rho_B)\Delta M}\\
    \le & L\sum_{M=0}^{\min(N,(1-\rho_A)\Delta)}\binom{N}{M}(2^{-(1-\rho_A)\Delta})^{N-M}(2^{-(1-\rho_B)\Delta})^M\\
    \le & L(2^{-(1-\rho_A)\Delta}+2^{-(1-\rho_B)\Delta})^N\\
    \le & L2^N 2^{-(1-\rho_A)\Delta N}\, ,
\end{align}
where we apply Lemma~\ref{lem:rank_prob} in the second line and also extend the limits of summation down to $M=0$ for convenience. We write
\begin{align}
    L = \max_{0\le M \le \min(N,(1-\rho_A)\Delta)}\left(2^{(N-M)M}\right)\, ,
\end{align}
which we extract from the sum in the third line above. We apply the binomial theorem in going to the fourth line, and the last line follows from the assumption that $\rho_B\le \rho_A$. To bound the remaining expression, we split into a two cases depending on the sizes of $(1-\rho_A)\Delta$ and $N$.

\begin{enumerate}
\item If we have $N \le 2(1-\rho_A)\Delta$, then $L = 2^{N^2/4} \le 2^{(1-\rho_A)\Delta N/2}$ and we have
\begin{align}
    L2^N2^{-(1-\rho_A)\Delta N} &\le 2^{(1-\rho_A)\Delta N/2+N-(1-\rho_A)\Delta N}\\ 
    &= 2^{-(1/2)(1-\rho_A)\Delta N + N}\\ 
    &= 2^{-\Theta(\Delta N)}\, .
\end{align}
\item If $N > 2(1-\rho_A)\Delta$, then $L = 2^{(1-\rho_A)\Delta(N-(1-\rho_A)\Delta)}$ and we have
\begin{align}
    L2^N2^{-(1-\rho_A)\Delta N} &= 2^{(1-\rho_A)\Delta N - ((1-\rho_A)\Delta)^2 + N -(1-\rho_A)\Delta N}\\ 
    &= 2^{-(1-\rho_A)^2\Delta^2 + N}\\ 
    &= 2^{-\Theta(\Delta^2)}\, .
\end{align}
\end{enumerate}
Since $N = \mathrm{rank}(X) > \Delta^{1/2+b}$, it follows that we have
\begin{align}
    \sum_{M=(\delta/2)\Delta^{1/2-2\varepsilon}}^{\mathrm{rank}(X)}\Pr_{H_A}(\mathrm{rank}(H_AX) = M)2^{-(1-\rho_B)\Delta M} = 2^{-\Omega(\Delta^{3/2+b})}
\end{align}
in both cases. Bounding $|\mathcal{X}_2|$ using Lemma~\ref{lem:lowrank_card_bound}, we finally get
\begin{align}
    P_2 \le |\mathcal{X}_2|2^{-\Omega(\Delta^{3/2+b})} \le 2^{\tilde{\Theta}(\Delta^{3/2+2\varepsilon})}2^{-\Omega(\Delta^{3/2+b})} = 2^{-\Omega(\Delta^{3/2+b})}\, ,
\end{align}
where the last equation follows from the fact that we chose $2\varepsilon < b$, so that $\Delta^{3/2+b}$ asymptotically dominates over $\Delta^{3/2+2\varepsilon}$.
\end{proof}

Altogether, combining the bounds for $P_1$ and $P_2$, it follows that
\begin{align}
    \Pr_{H_A,H_B}(C_{AB} \text{ is not SR} \text{ and } d \ge \delta\Delta) \le P_1 + P_2 \le 2^{-\Theta(\Delta^{3/2-2\varepsilon})} + 2^{-\Omega(\Delta^{3/2+2\varepsilon})} = 2^{-\Theta(\Delta^{3/2-2\varepsilon})}\, .
\end{align}
\end{proof}

Theorem~\ref{thm:sparse_robust} shows that random dual tensor codes are sparse robust with high probability. We now proceed to use this result to show that random dual tensor codes are also sparse robust with respect to puncturing with high probability. The main result of this section is the following theorem.

\begin{theorem}[Sparse Robustness with respect to Puncturing]\label{thm:puncture_robust}
Fix constants $\rho_A,\rho_B \in (0,1)$, $\varepsilon \in (0,1/14)$, and $\delta \in (0,1/2)$ with $\delta < \min(h^{-1}(\rho_A),h^{-1}(\rho_B))$, where $h(x)$ is the binary entropy function.

Let $H_A\in \mathbb F_2^{(1-\rho_A)\Delta\times \Delta}$ and $H_B\in \mathbb F_2^{(1-\rho_B)\Delta\times\Delta}$ be uniformly random binary check matrices defining codes $C_A$ and $C_B$ respectively. Then $C_{AB}$ has distance $d\ge \delta\Delta$ and is sparse robust with respect to puncturing with high probability. More precisely, we have
\begin{align}
    \Pr_{H_A,H_B}(C_{AB} \text{ is SRP} \text{ and } d \ge \delta\Delta) \ge 1 - 2^{-\Omega(\Delta)}\, .
\end{align}
In particular, it follows from Lemma~\ref{lem:sparse_implies_robust} that random dual tensor codes have distance $d\ge \delta\Delta$ and are sufficiently robust with high probability.
\end{theorem}
\begin{proof}
We have
\begin{align}
    \Pr_{H_A,H_B}(C_{AB} \text{ is SRP} \text{ and } d \ge \delta\Delta) = 1 - \Pr_{H_A,H_B}(C_{AB} \text{ is not SRP} \text{ or } d < \delta\Delta)\, .
\end{align}
We will upper bound the latter probability. For $\delta < \min(h^{-1}(\rho_A),h^{-1}(\rho_B))$, the Gilbert-Varshamov bound implies that randomly chosen parity check matrices $H_A, H_B$ will define codes with minimum distances $d = \min(d_A,d_B) \ge \delta\Delta$ with probability $1-2^{-\Omega(\Delta)}$. Taking a union bound, we have
\begin{align}
    \Pr_{H_A,H_B}(C_{AB} \text{ is not SRP} \text{ or } d < \delta\Delta) &\le \Pr_{H_A,H_B}(C_{AB} \text{ is not SRP} \text{ and } d \ge \delta\Delta) + \Pr_{H_A,H_B}(d < \delta\Delta)\\
    &=\Pr_{H_A,H_B}(C_{AB} \text{ is not SRP} \text{ and } d \ge \delta\Delta) +  2^{-\Omega(\Delta)}\, .
\end{align}

Let $\mathcal{A}'$ and $\mathcal{B}'$ be the set of all coordinates obtained from $A$ and $B$ by puncturing on a subset of size $\Delta^{1-\varepsilon}$. Note that $C_{AB}$ will fail to be SRP if and only if there exists some $A'\in \mathcal{A}'$ and $B' \in \mathcal{B}'$ such that the punctured code $C_{A'B'}$ is not SR. We can therefore take a union bound over $\mathcal{A}'$ and $\mathcal{B}'$ to get  
\begin{align}\label{eq:Pr_CAB_not_SRP}
    \Pr_{H_A,H_B}\left(C_{AB} \text{ is not SRP and }d\ge \delta\Delta\right) \le \sum_{\substack{A'\in \mathcal{A}'\\B'\in\mathcal{B}'}}\Pr_{H_A,H_B}(C_{A'B'} \text{ is not SR and } d\ge \delta\Delta)\, .
\end{align}
To handle the puncturing, it is more convenient to take the random codes over uniformly chosen generator matrices. To that end, we can apply Corollary~\ref{cor:total_var_bound} to get
\begin{align}
    \Pr_{H_A,H_B}(C_{A'B'} \text{ is not SR and } d\ge \delta\Delta) \le \Pr_{G_A,G_B}(C_{A'B'} \text{ is not SR and } d\ge \delta\Delta) + 2^{-\Omega(\Delta)}\, ,
\end{align}
where the latter probability is over codes defined by randomly chosen generator matrices (of the appropriate sizes). Since $G_A$ and $G_B$ are chosen uniformly randomly, it follows that the the generator matrices for their punctured codes $G_{A'}$ and $G_{B'}$ are also chosen uniformly randomly. Since we only puncture on a sublinear number of entries, the distance $d'$ of the punctured code is guaranteed to be above, say $0.9\delta\Delta$, for sufficiently large $\Delta$. Therefore we have
\begin{align}
    \Pr_{G_A,G_B}(C_{A'B'} \text{ is not SR and } d\ge \delta\Delta) &\le\Pr_{G_A,G_B}(C_{A'B'} \text{ is not SR and } d'\ge 0.9\delta\Delta)\\
    &= \Pr_{G_{A'},G_{B'}}(C_{A'B'} \text{ is not SR and } d'\ge 0.9\delta\Delta)\\
    &\le \Pr_{H_{A'},H_{B'}}(C_{A'B'} \text{ is not SR and } d'\ge 0.9\delta\Delta) + 2^{-\Omega(\Delta)}\, ,
\end{align}
where in the last line we apply Corollary~\ref{cor:total_var_bound} once again to return to the distribution over uniform check matrices $H_{A'}$ and $H_{B'}$.
We can now apply Theorem~\ref{thm:sparse_robust} with our chosen parameters\footnote{Note that the blocklength of the punctured code is proportional to $\Delta' = \Delta - \Delta^{1-\varepsilon}$. Since the leading order behavior is unchanged, we have $\Theta(\Delta^{3/2-2\varepsilon}) = \Theta((\Delta')^{3/2-2\varepsilon})$.} to conclude that
\begin{align}
    \Pr_{H_{A'},H_{B'}}(C_{A'B'} \text{ is not SR and } d'\ge 0.9\delta\Delta) \le 2^{-\Theta(\Delta^{3/2-2\varepsilon})\, }.
\end{align}
It remains to bound the sizes of $\mathcal{A}'$ and $\mathcal{B}'$.
There are at most
\begin{align}
    \binom{\Delta}{\Delta^{1-\varepsilon}} = 2^{\tilde{\Theta}(\Delta^{1-\varepsilon})} 
\end{align}
ways to puncture $\Delta^{1-\varepsilon}$ coordinates of $A$ (or $B$). Therefore we get $|\mathcal{A}'|\cdot |\mathcal{B}'| = 2^{\tilde{\Theta}(\Delta^{1-\varepsilon})}$. Returning to \eqref{eq:Pr_CAB_not_SRP}, we have the following bound of
\begin{align}
    \Pr_{H_A,H_B}\left(C_{AB} \text{ is not SRP and }d\ge \delta\Delta\right) &\le |\mathcal{A}'|\cdot |\mathcal{B}'| \cdot (2^{-\Theta(\Delta^{3/2-2\varepsilon})}+2^{-\Omega(\Delta)})\\
    &= 2^{\tilde{\Theta}(\Delta^{1-\varepsilon})}2^{-\Omega(\Delta)}\\ 
    &= 2^{-\Omega(\Delta)}\, .
\end{align}
Therefore
\begin{align}
    \Pr_{H_A,H_B}(C_{AB} \text{ is SRP} \text{ and } d \ge \delta\Delta) &= 1 - \Pr_{H_A,H_B}(C_{AB} \text{ is not SRP} \text{ or } d < \delta\Delta)\\
    &\ge 1 - 2^{-\Omega(\Delta)}\, .
\end{align}
and the result follows.
\end{proof}

Theorem~\ref{thm:dualtensorrobustness} follows easily from Theorem~\ref{thm:sparse_robust} and Lemma~\ref{lem:sparse_implies_robust}.

\begingroup
\def\thetheorem{\ref{thm:dualtensorrobustness}}
\begin{theorem}
Fix constants $\varepsilon\in (0,1/28)$, $\rho\in (0,1/2)$, and $\delta\in (0,1/2)$ such that $\delta<h^{-1}(\rho)$, where $h(x)$ is the binary entropy function. For all sufficiently large $\Delta$, there exist classical codes $C_A,C_B$ of length $\Delta$ and rates $\rho_A= \rho$ and $\rho_B=1-\rho$ such that such that both the dual tensor code of $C_A$ and $C_B$ and the dual tensor code of $C_A^\perp$ and $C_B^\perp$ are $\Delta^{3/2+\varepsilon}$-robust and have distances at least $\delta\Delta$.
\end{theorem}
\addtocounter{theorem}{-1}
\endgroup

\begin{proof}
Let $C_A$ be a uniformly random classical code of length $\Delta$ and rate $\rho$. That is, $C_A$ is a uniformly random $\rho\Delta$-dimensional subspace of $\bb F_2^\Delta$. Similarly, let $C_B$ be a random $(1-\rho)\Delta$-dimensional subspace of $\bb F_2^\Delta$. By Theorem~\ref{thm:sparse_robust} and Lemma~\ref{lem:sparse_implies_robust}, we have
\begin{equation}
    \Pr_{C_A,C_B}(C_{AB} \text{ is not $\Delta^{3/2+\varepsilon}$-robust} \text{ or } d < \delta\Delta)\le 2^{-\Omega(\Delta)}\, ,
\end{equation}
where we also use Corollary~\ref{cor:total_var_bound} to switch from the distribution defined by random parity check matrices to one defined by random subspaces. Since $C_A^\perp$ and $C_B^\perp$ are also uniformly random subspaces of $\bb F_2^\Delta$ of dimensions $(1-\rho)\Delta$ and $\rho\Delta$ respectively, we also have
\begin{equation}
    \Pr_{C_A,C_B}(C_{A^\perp B^\perp} \text{ is not $\Delta^{3/2+\varepsilon}$-robust} \text{ or } d^\perp < \delta\Delta)\le 2^{-\Omega(\Delta)}\, ,
\end{equation}
where $C_{A^\perp B^\perp}$ is the dual tensor code of $C_A^\perp$ and $C_B^\perp$ and $d^\perp$ is the distance of $C_{A^\perp B^\perp}$. Therefore,
\begin{equation}
    \Pr_{C_A,C_B}(C_{AB} \text{ and } C_{A^\perp B^\perp} \text{ are $\Delta^{3/2+\varepsilon}$-robust and } d,d^\perp\ge \delta\Delta ) \ge 1-2^{-\Omega(\Delta)}\, ,
\end{equation}
so for sufficiently large $\Delta$, there exist $C_A,C_B$ satisfying the conditions. Note that we require $\varepsilon<1/28$ in the theorem because the SRP parameter of up to $1/14$ in Theorem~\ref{thm:sparse_robust} is halved in Lemma~\ref{lem:sparse_implies_robust}.
\end{proof}

We remark that we did not give the tightest bounds in the section because in the proof of our decoder, we only needed dual tensor codes with $\Delta^{3/2+ \varepsilon}$-robustness for any $\varepsilon>0$. By more carefully tracking the exponents throughout the argument, it is possible to show the existence of $\Delta^{3/2+\varepsilon}$-robust dual tensor codes for any $\varepsilon<1/6$.

\end{document}